\documentclass[journal]{IEEEtran}
\usepackage{amsmath,amsfonts,amssymb,dsfont,mathrsfs}
\usepackage{amsthm}

\newtheorem{proposition}{Proposition}
\newtheorem{lemma}{Lemma}
\usepackage{algorithm, algorithmicx, algpseudocode}
\usepackage{array}
\usepackage{textcomp}
\usepackage{stfloats}
\usepackage{url}
\usepackage{verbatim}
\usepackage{graphicx}
\usepackage{cite}
\usepackage{enumerate}
\usepackage{bm}
\usepackage{makecell}
\usepackage{balance}
\usepackage{booktabs}
\usepackage{cuted}
\usepackage{xcolor}
\usepackage{multirow}
\usepackage{bbding}
\usepackage[caption=false,font=footnotesize,labelfont=rm,textfont=rm]{subfig}
\usepackage{tablefootnote}
\usepackage{ragged2e}
\usepackage{placeins}
\usepackage{textcomp}
\usepackage{float}
\usepackage{afterpage}

\newcolumntype{L}[1]{>{\RaggedRight\arraybackslash}p{#1}}
\newcolumntype{C}[1]{>{\centering\arraybackslash}p{#1}}

\usepackage[colorlinks,
            linkcolor=blue,       
            anchorcolor=blue,  
            citecolor=blue,        
            ]{hyperref}
\def\BibTeX{{\rm B\kern-.05em{\sc i\kern-.025em b}\kern-.08em
    T\kern-.1667em\lower.7ex\hbox{E}\kern-.125emX}}
\usepackage{lineno} 

\def\gap{0.53ex}
\abovedisplayskip\gap
\belowdisplayskip\gap
\abovedisplayshortskip\gap
\belowdisplayshortskip\gap

\begin{document}

\title{SCAN-BEST: Sub-6GHz-Aided Near-field Beam Selection with Formal Reliability Guarantees}
\author{Weicao Deng,~\IEEEmembership{Graduate Student Member,~IEEE},
Binpu Shi,~\IEEEmembership{Graduate Student Member,~IEEE},\\
Min Li,~\IEEEmembership{Member,~IEEE},
and Osvaldo Simeone,~\IEEEmembership{Fellow,~IEEE}
\thanks{
The work of Binpu Shi was supported by a Project Supported by Scientific Research Fund of Zhejiang University (No. XY2024007).
The work of Osvaldo Simeone was supported by the European Research Council (ERC) under the European Union’s Horizon Europe Programme (grant agreement No. 101198347), by an Open Fellowship of the EPSRC (EP/W024101/1), and by the EPSRC project (EP/X011852/1).

Weicao Deng, Binpu Shi, and Min Li are with the College of Information Science and Electronic Engineering and also with Zhejiang Provincial Key Laboratory of Multi-Modal Communication Networks and Intelligent Information Processing, Zhejiang University, Hangzhou 310027, China (e-mail: \{caowd,~bp.shi,~min.li\}@zju.edu.cn). (\emph{Corresponding author: Min Li.})

Osvaldo Simeone is with the Intelligent Networked Systems Institute (INSI), Northeastern University London, One Portsoken Street, London E1 8PH, United Kingdom (email: o.simeone@northeastern.edu). 
}}

\maketitle

\begin{abstract}
As millimeter-wave (mmWave) MIMO systems adopt larger antenna arrays, near-field propagation becomes increasingly prominent, especially for users close to the transmitter. 
Traditional far-field beam training methods become inadequate, while near-field training faces the challenge of large codebooks due to the need to resolve both angular and distance domains. 
To reduce in-band training overhead, prior work has proposed to leverage the spatial-temporal congruence between sub-6 GHz (sub-6G) and mmWave channels to predict the best mmWave beam within a near-field codebook from sub-6G channel estimates.
To cope with the uncertainty caused by sub-6G/mmWave differences, we introduce a novel \textcolor{black}{Sub-6G Channel Aided Near-field BEam SelecTion (SCAN-BEST) framework} that wraps around any beam predictor to produce candidate beam subset with formal suboptimality guarantees.
The proposed SCAN-BEST builds on conformal risk control (CRC), and is calibrated offline using limited calibration data.
Its performance guarantees apply even in the presence of statistical shifts between calibration and deployment.
Numerical results validate the theoretical properties and efficiency of SCAN-BEST.
\end{abstract}


\section{Introduction}
Millimeter-wave (mmWave) and extremely large-scale massive multiple-input multiple-output (XL-MIMO) are recognized as key enablers for 6G communication systems \cite{8454665,10729214}.
However, their adoption introduces unique challenges due to the high operating frequencies and to the large antenna array sizes, which result in a substantial increase in the Rayleigh distance—from a few meters to several hundred meters \cite{cui2022near}. 
This significantly expands the near-field region, where electromagnetic waves must be accurately modeled using spherical wavefronts rather than planar approximations \cite{10683443}. 
Communication quality in this regime becomes highly susceptible to blockages \cite{liu2024near}, presenting significant challenges for channel state information (CSI) acquisition and beamforming design.

\par This work explores the idea of using sub-6 GHz (sub-6G) channel data to enhance mmWave communication.
As illustrated in Fig.~\ref{temporal spatial congruence}, sub-6G and mmWave channels exhibit similar power delay profiles (PDPs) \cite{samimi20163,10292615}.
This congruence suggests the potential to extract both angular and distance information from sub-6G channels to facilitate near-field beam selection.
Nonetheless, significant challenges remain due to the inherent differences between sub-6G and mmWave channels, such as variations in path parameters and discrepancies in angular and temporal resolutions. 
These differences introduce uncertainties when mapping sub-6G channel data to mmWave beams, as illustrated by the dominant path differences in Fig.~\ref{temporal spatial congruence:c}. 
Additionally, ensuring guaranteed performance for near-field beam selection remains an open research question and a critical challenge.

\subsection{Related Works and Motivations} To address these challenges, a significant number of works \cite{cuiChannelEstimationExtremely2022,zhangFastNearFieldBeam2022,10239282,10163797,9903646} have focused on the angle and distance characteristics of near-field channels, proposing various beam training approaches.
For instance, reference \cite{cuiChannelEstimationExtremely2022} introduced on-grid/off-grid polar-domain simultaneous orthogonal matching pursuit (OMP) channel estimation methods that exploit the polar-domain sparsity of near-field channels. 
Reference \cite{zhangFastNearFieldBeam2022} developed a two-phase beam training scheme that divides the two-dimensional search in the polar-domain codebook into two sequential angular and distance domain searches. 
Similarly, reference \cite{10239282} proposed a two-stage hierarchical beam training method.
In the first stage, the central sub-array is used to perform a coarse search for the user direction in the angular domain. 
In the second stage, a fine-grained search for the user’s direction and distance is conducted in the polar domain.

\par Inspired by the integration of wireless communication and artificial intelligence (AI) \cite{10431795}, reference \cite{10163797} proposed frameworks to predict the optimal angle and distance using extensive near-field beam training.
Likewise, reference \cite{9903646} developed a framework that leverages received signals from the far-field wide beam training.
These studies rely solely on in-band measurements and often struggle to achieve optimal performance under a limited pilot budget or low signal-to-noise ratio (SNR) conditions.

\begin{figure}[t]
\centering
\subfloat[]{
\begin{minipage}[b]{0.24\textwidth}
    \includegraphics[width=1\textwidth]{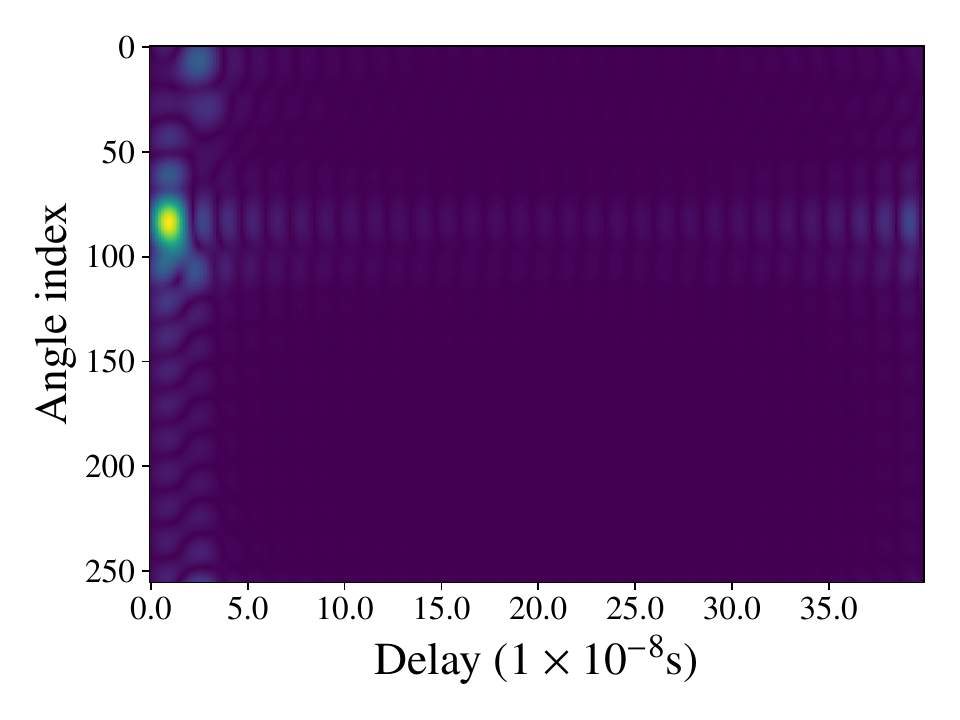}\\
    \includegraphics[width=1\textwidth]{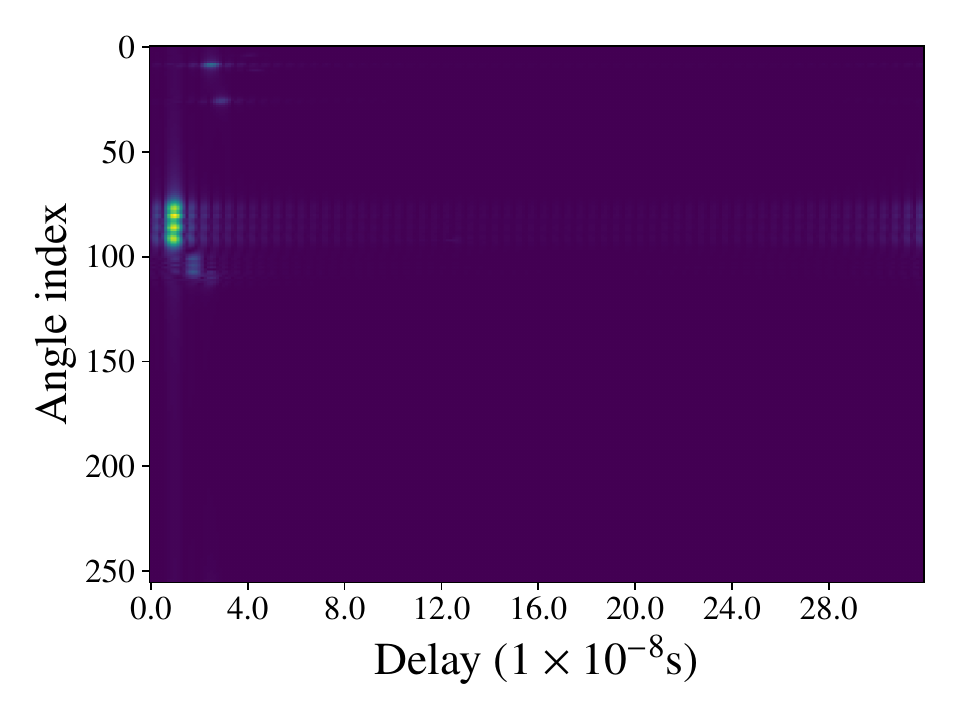}\label{temporal spatial congruence:a}
\end{minipage}}
\subfloat[]{
\begin{minipage}[b]{0.24\textwidth}
    \includegraphics[width=1\textwidth]{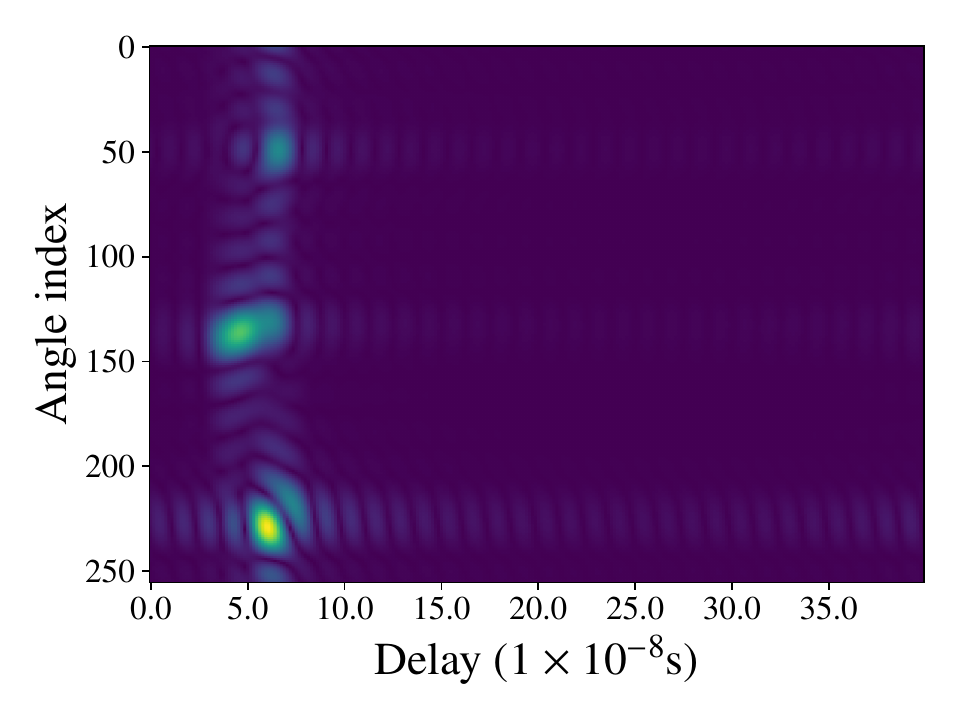}\\
    \includegraphics[width=1\textwidth]{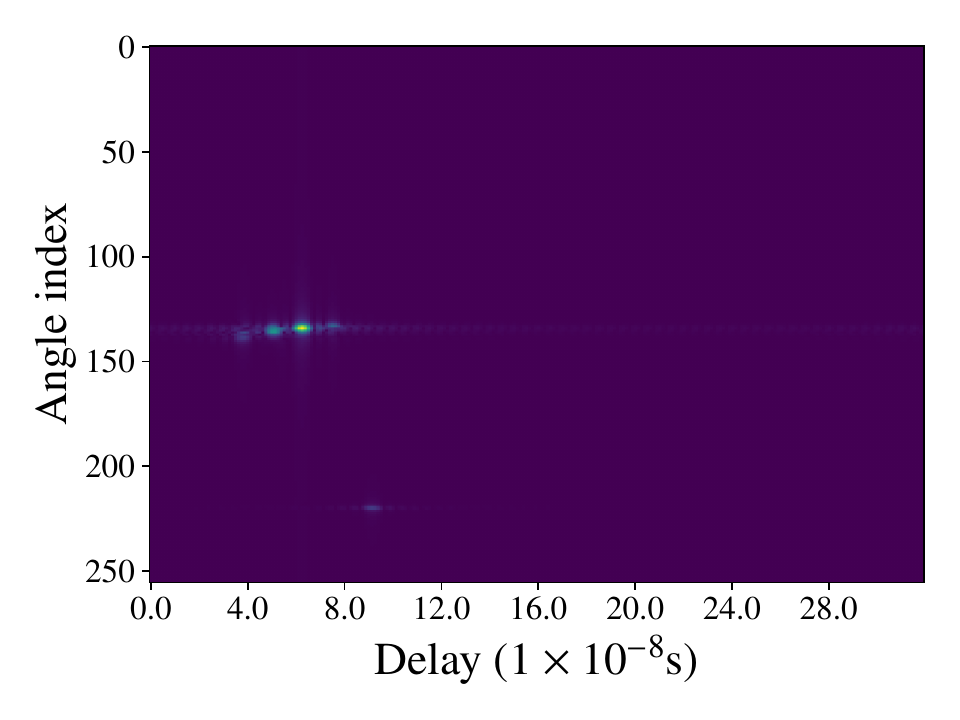}\label{temporal spatial congruence:c}
\end{minipage}}
\caption{An illustration of spatial and temporal correlation between sub-6G and mmWave channels. Each column represents a pair of sub-6G (upper row) and mmWave (lower row) augmented angle and delay profiles. 
The first column describes a line-of-sight (LoS) condition, while the second column represents a non-LoS (NLoS) condition \cite{Remcom}.}
\label{temporal spatial congruence}
\end{figure}

\par Several studies \cite{ali2017millimeter,alrabeiah2020deep,dengwcnc2024, deng2024csi} have utilized the angular domain congruence between sub-6G and mmWave channels to assist the far-field beam selection or beamforming.
Building on \cite{ali2017millimeter}, reference \cite{wu2024near} proposed a complex simultaneous logit-weighted block OMP algorithm for near-field channel estimation, which leverages the angular domain congruence between sub-6G and mmWave channel to assign weights.
However, in the context of near-field beamforming or channel estimation, it's not just the angle information that is requisite; distance information is similarly necessary.
Moreover, the inherent channel variances between sub-6G and mmWave frequencies present a significant challenge to reliably mapping sub-6G information to mmWave beam with performance guarantee.

\par Conformal risk control (CRC) is a recent technique that aims at providing risk guarantees for the set-based decision, which has been extensively studied in numerous previous works in mathematics and statistics \cite{kuchibhotla2020exchangeability,angelopoulos2022conformal}, been applied to a lot of engineering settings including computer vision, natural language processing, information retrieval systems \cite{xu2024two}, health imaging \cite{hulsman2024conformal}, anomaly detection \cite{bai2025crc}, and telecommunication \cite{10445122,zhu2024conformal}.
Extensions of CRC include cross-validation-based method \cite{cohen2024cross} and online strategies \cite{angelopoulos2023conformal,zecchin2024localized}.

\par In the context of wireless communication, reference \cite{10262367} was the first to explore the application of conformal prediction (CP), a precursor of CRC \cite{angelopoulos2022conformal}, to the design of AI for communication systems, with focus on demodulation, modulation classification, and channel prediction.
Additionally, reference \cite{10416237} investigated federated CP in a wireless setting, proposing a novel wireless federated CP framework for federated reliable inference, and the work \cite{wen2025distributed} extended the approach to a fully decentralized setting.
Inspired by the prior works, this work further explores the use of CRC to enable reliable near-field beam selection.

\subsection{Main Contributions}
The correlation between sub-6G and mmWave channels illustrated in Fig.~\ref{temporal spatial congruence} offers promising potential for enabling low-overhead near-field beam selection, surpassing traditional in-band schemes. 
However, the non-negligible discrepancies between these two frequency bands introduce substantial uncertainty in predicting mmWave near-field beams based on sub-6G information.
To address this challenge, we develop Sub-6GHz Channel Aided Near-field BEam SelecTion (SCAN-BEST), a framework combining deep learning for beam prediction with CRC to ensure reliability guarantee.
Our main contributions are summarized as follows: 
\begin{enumerate}
    \item We propose the SCAN-BEST, a framework for reliable sub-6 GHz channel-aided near-field beam prediction. 
    SCAN-BEST is designed to wrap around any near-field beam predictor, providing a systematic mechanism to quantify and control predictive uncertainty in beam selection — an issue largely overlooked in prior works. 
    Specifically, through a CRC-based scheme, SCAN-BEST efficiently constructs a near-field candidate beam subset that formally meets a target suboptimality ratio with respect to the best beam in this set with a prescribed probability. 
    Furthermore, we extend SCAN-BEST to handle statistical distribution shifts between calibration and test data by incorporating a weighted CRC variant \cite{tibshirani2019conformal}. 
    This extension ensures that SCAN-BEST remains reliable even when there are discrepancies between the environments encountered during calibration and deployment.
    
    \item We conduct comprehensive numerical simulations to validate the efficiency of SCAN-BEST by comparing it with various sub-6G-based and mmWave baselines. 
    The performance is evaluated across different target suboptimal rates and calibration dataset sizes. 
    We also examine the scalability of SCAN-BEST by varying the quality of sub-6G data, including changes in the number of sub-6G antennas and the power levels in sub-6G channel estimation.
    Moreover, we validate the performance of SCAN-BEST in the presence of the covariate shift between the calibration and test data via weighted CRC.

\end{enumerate}

\par The remainder of this paper is organized as follows. 
Section~\ref{sec:system} describes the system model and formulates the problem of sub-6G channel-aided near-field beam selection as well as its calibration.
Section~\ref{sec:method} elaborates on the proposed SCAN-BEST framework from its implementation.
In Section~\ref{sec:numerical results}, we present and discuss the numerical results to validate the effectiveness of SCAN-BEST.
Section~\ref{sec:conclusion} provides conclusions for this paper.

\subsubsection*{Notations} Scalars, vectors and matrices are respectively denoted by lower/upper case, boldface lower case and boldface upper case letters.
Notation ${\mathbf I}_a$ represents an $a\times a$ identity matrix.
$\mathcal{CN}(0,\sigma^2)$ is a zero-mean complex Gaussian distribution with variance $\sigma^2$.
The function $\text{card}(\mathcal{X})$ returns the cardinality of set $\mathcal{X}$.
The notation $A\times B$ is also used for the Cartesian product of the sets $\{1, 2, ..., A\}$ and $\{1,2, ..., B\}$ with integers $A$ and $B$.
Moreover, to distinguish between the sub-6G system and mmWave system, we use $\underline{(\cdot)}$ to indicate parameters corresponding to the sub-6G system, as exemplified by $\underline{x}$.

\section{System Model and Problem Formulation}
\label{sec:system}
\begin{figure*}
\centering
\includegraphics[width=0.98\textwidth]{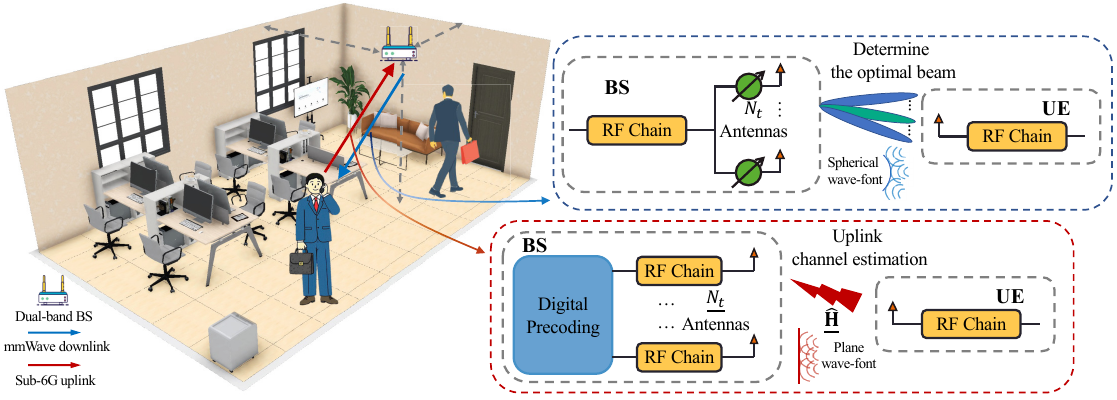}
\caption{An illustration of system model and beamforming architecture.}
\label{bs_ue_location}
\end{figure*}
As depicted in Fig.~\ref{bs_ue_location}, as in  \cite{ali2017millimeter,alrabeiah2020deep,dengwcnc2024}, we consider a dual-band system comprising a BS and single UE.
Both the BS and UE are equipped with two transceivers that operate in sub-6G frequencies and mmWave frequencies, respectively.
The UE is assumed to have a single antenna in both mmWave and sub-6G frequencies. 
As for the BS, the mmWave system is equipped \textcolor{black}{a uniform linear array (ULA) comprising $N_{\text{t}}$ antennas with half-wavelength spacing}, and it adopts a fully analog beamforming architecture.
In contrast, its sub-6G system is equipped with \textcolor{black}{a ULA of $\underline{N_{\text{t}}}$ antennas with half-wavelength spacing}, and it employs a fully digital beamforming architecture.
In the rest of this section, we will define the mmWave and sub-6G radio interfaces, as well as the problem definition.

\subsection{MmWave Downlink Communication and Channel Model}
We consider a wideband mmWave orthogonal frequency division multiplex (OFDM) system operating at the center frequency $f_{\text{c}}$ with a total system bandwidth $W$ and $M$ subcarriers. 
The sampling period $T_{\text{s}}$ is given by $T_{\text{s}} = 1/W$, and the number of channel taps at the resolution $T_{\text{s}}$ is denoted as $D$.
The signal transmitted at the $m$-th subcarrier $s_{\text{t},m}\in {\mathbb C}$ follows the complex Gaussian distribution $\mathcal{CN}(0,{P_{\text{t}}}/{M})$ with $P_{\text{t}}$ being the total transmitting power. 
Let $\mathbf{f}$ be the analog beam applied by the BS.
The received signal of the UE at the $m$-th subcarrier $y_m\in {\mathbb C}$ can be written as \cite{10643599}
\begin{equation}
y_m={\mathbf h}_k^H{\mathbf f}s_{\text{t},m}+n_m,
\end{equation}
where $\mathbf{h}_m^H\in {\mathbb C}^{1\times N_{\text{t}}}$ denotes the frequency-domain channel from the BS to the UE at the $k$-th subcarrier, and $n_m\sim \mathcal{CN}(0,\sigma_{\text{n}}^2)$ is the additive noise with the variance $\sigma_{\text{n}}^2$. 

\par We adopt a wideband geometric channel model consisting of a total of $L$ paths between the BS and the UE.
To account for large arrays, we model the extended near-field range via spherical wave-fronts \cite{cui2022near}.
The channel gain, angle of departure (AoD), and time of arrival (ToA) of the $l$-th path between the BS and the UE are denoted by $g_{l}$, $\theta_{l}$, and $\tau_{l}$, respectively.
Futuremore, if the $l$-th path is in line of sight (LOS), the additional path parameter $r_{l}$ is used to denote the distance between the BS and the UE, while otherwise, $r_{l}$ denotes the distance between the BS and the scatterer in the $l$-th path of the UE.
Overall, the delay-$d$ mmWave channel tap from the BS to UE, denoted as ${\mathbf h}_d^H\in {\mathbb C}^{1\times N_{\text{t}}}$, is given by \cite{ali2017millimeter}
\begin{equation}
{\mathbf h}_d^H=\sqrt{N_{\text{t}}}\sum_{l=1}^{L}g_{l} p(dT_{\text{s}}-\tau_{l}){\mathbf b}^H(\theta_{l},r_{l}), \label{tap d mmWave channel}
\end{equation}
for $d=1,\dots,D$, where $p(\cdot)$ denote the pulse-shaping filter, and the steering vector ${\mathbf b}^H(\theta_{l},r_{l})$ is given by
\begin{equation}
{\mathbf b}^H(\theta_{l},r_{l})=\frac{1}{\sqrt{N_{\text{t}}}}[e^{-j2\pi(r_{l}^{(0)}-r_{l})},\dots,e^{-j2\pi(r_{l}^{(N_{\text{t}}-1)}-r_{l})}],
\end{equation}
with $r_{l}^{(n)}=\sqrt{r_{l}^2+{\delta^{(n)}}^2(\lambda/2)^2-r_{l}\theta_{l}\delta^{(n)} \lambda}$, with $\delta^{(n)}={(2n-N_{\text{t}}-1)}/{2}$ \cite{cuiChannelEstimationExtremely2022}.
Based on the delay-$d$ channel tap in \eqref{tap d mmWave channel}, the frequency-domain mmWave channel at the $m$-th subcarrier, ${\mathbf h}_m^H\in {\mathbb C}^{1\times N_{\text{t}}}$ is given by
\begin{equation}
{\mathbf h}_m^H=\sum_{d=1}^{D}{\mathbf h}_d^H e^{-j\frac{2\pi m}{M}d},
\end{equation}
for $m=1,\dots,M$.
For notational convenience, we denote the overall frequency-domain mmWave OFDM channel as ${\mathbf H}\triangleq[{\mathbf h}_1,\dots,{\mathbf h}_M]^H$.

\subsection{Sub-6G Channel Model and Uplink Channel Estimation}
The sub-6G link between BS and UE operates at the center frequency $\underline{f_{\text{c}}}$, with a total bandwidth $\underline{W}$, number of subcarriers $\underline{M}$, sampling period $\underline{T_{\text{s}}} = 1/\underline{W}$, and number of channel taps $\underline{D}$.
Due to the more abundant scattering and diffraction of sub-6G propagation with respect to (w.r.t.) the mmWave band, there typically exist $\underline{L}> L$ paths between the BS and UE \cite{samimi20163,ali2017millimeter}.
We denote the gain, angle of arrival (AoA), time of arrival (ToA) of the $l$-th channel path with $\underline{g}_{l}$, $\underline{\theta}_{l}$, and $\underline{\tau}_{l}$, respectively.
Due to the shorter Rayleigh distance of the sub-6G system compared to the mmWave system, the UE is likely to be located in either the far-field or near-field region of the sub-6G system.
Thus, similar to the latest work \cite{wu2024near}, the sub-6G channel between the BS and the UE is modeled under two cases.

\par When the UE is located in the far-field region, the delay-$d$ sub-6G channel tap from the UE to BS $\underline{{\mathbf h}_{d}}\in {\mathbb C}^{\underline{N_{\text{t}}}\times 1}$ is given by \cite{ali2017millimeter}
\begin{equation}
\underline{{\mathbf h}}_{d}=\sqrt{\underline{N_{\text{t}}}}\sum_{l=1}^{\underline{L}}\underline{g}_{l}p(d\underline{T_{\text{s}}}-\underline{\tau}_l)\underline{\mathbf a}(\underline{\theta}_l), \label{tap d sub-6G channel}
\end{equation}
for $d=1,\dots,\underline{D}$, where the steering vector $\underline{\mathbf a}(\underline{\theta}_l)$ is
\begin{equation}
\underline{\mathbf a}(\underline{\theta}_l)= \frac{1}{\sqrt{\underline{N_{\text{t}}}}}[1,e^{-j\pi \sin(\underline{\theta}_l)},\dots,e^{-j\pi(\underline{N_{\text{t}}}-1)\sin(\underline{\theta}_l)}]^T.
\end{equation}
When the UE is located in the near-field region, $\underline{{\mathbf h}_{d}}$ is similar to the mmWave near-field one and thus will not be elaborated further here.

\par Using \eqref{tap d sub-6G channel}, the frequency-domain sub-6G channel at the $m$-th subcarrier $\underline{{\mathbf h}}_m\in {\mathbb C}^{\underline{N_{\text{t}}}\times 1}$ is given by
\begin{equation}
\underline{{\mathbf h}}_m=\sum_{d=1}^{\underline{D}}\underline{{\mathbf h}}_d e^{-j\frac{2\pi m}{\underline{M}}d},
\end{equation}
for $m=1,\dots,\underline{M}$.

\par Through a fully digital receiver architecture in the sub-6G system, the frequency-domain channels at the $\underline{M}$ subcarriers is estimated via sub-6G pilot transmission.
To this end, the UE sends the uplink pilot signal $\underline{s}_{\text{p},m}=\sqrt{{\underline{P_{\text{s}}}}/{\underline{M}}}$ at the $m$-th subcarrier with $\underline{P_{\text{s}}}$ being the total power of the pilot signal.
The received signals by the BS at the $m$-th subcarrier $\underline{{\mathbf y}}_{\text{p},m}$ can be written as 
\begin{equation}
\underline{{\mathbf y}}_{\text{p},m}=\underline{{\mathbf h}}_m\underline{s}_{\text{p},m}+\underline{{\mathbf n}}_m,
\end{equation}
where $\underline{\mathbf n}_m\sim \mathcal{CN}(0,\underline{\sigma_{\text{n}}^2}{\mathbf I}_{\underline{N_{\text{t}}}})$ denotes the noise vector with $\underline{\sigma_{\text{n}}^2}$ being the sub-6G noise power.
The frequency-domain sub-6G channel estimate at the $m$-th subcarrier $\underline{\widehat{\mathbf h}}_m$ can be obtained via a low-complexity least square (LS) algorithm, which is given by
\begin{equation}
\underline{\widehat{\mathbf h}}_m=\underline{{\mathbf y}}_{\text{p},m}~\underline{s}_{\text{p},m}^{-1}.
\end{equation} 
We denote the frequency-domain sub-6G OFDM channel estimate as $\widehat{\underline{\mathbf H}}\triangleq [~{\underline{\widehat{\mathbf h}}_1},\dots,{\underline{\widehat{\mathbf h}}_{\underline{M}}}~]^H$.

\subsection{Problem Formulation}
Following prior works \cite{zhangFastNearFieldBeam2022,pengChannelEstimationExtremely2023}, we consider the problem of beam selection from a predefined near-field codebook.
Specifically, we adopt the polar codebook introduced in \cite{cuiChannelEstimationExtremely2022}, denoted as $\mathcal{W}\triangleq\left(\mathbf{w}_{n,s}\right)_{N_{\text{t}}\times S}$, where the numbers of candidate angles and distances are equal to the number of mmWave transmit antennas, $N_{\text{t}}$, and to an integer parameter $S$, respectively.
Each beam ${\mathbf w}_{n,s}={\mathbf b}(\theta_n, r_{n,s})$ in the set $\mathcal{W}$ corresponds to the $(n,s)$-th angle-distance sector, with
\begin{align}
    \theta_n &=\arcsin\left(\frac{2n-N_{\text{t}}-1}{N_{\text{t}}}\right),\notag\\
    r_{n,s}&=\frac{\left(1-\sin^2(\theta_n)\right)N_{\text{t}}^2d^2}{2s\beta^2\lambda},
\end{align}
where $\beta$ is the correlation parameter between neighboring beams \cite{cuiChannelEstimationExtremely2022}.

\par For a given mmWave channel sample $\mathbf{H}$, the \textit{optimal beam} $\mathbf{f}^\star$ is defined as the beam in set $\mathcal{W}$ that maximizes the average spectral efficiency, i.e., as
\begin{equation}
{\mathbf f}^{\star}={\mathbf w}_{n^{\star},s^{\star}}=\mathop{\arg\max}\limits_{{\mathbf w}_{n,s}\in {\mathcal W}}~R(\mathbf{w}_{n,s},\mathbf{H}), \label{eq: optimal beam}
\end{equation}
where the average spectral efficiency is over all subcarriers given by
\begin{equation}
R(\mathbf{w}_{n,s},\mathbf{H})=\frac{1}{M}\sum_{m=1}^{M}\log_2\left(1+\frac{\frac{P_{\text{t}}}{M}|{\mathbf h}_m^H{\mathbf w}_{n,s}|^2}{\sigma_{\text{n}}^2}\right).\label{eq: average spectral efficiency}
\end{equation}

\par Most previous works \cite{cuiChannelEstimationExtremely2022,zhangFastNearFieldBeam2022,10239282,10163797,9903646} have focused on in-band beam training to determine the optimal near-field beam \eqref{eq: optimal beam}.
However, these methods typically require a substantial number of pilots.
As in \cite{alrabeiah2020deep,liuNMBEnetEfficientNearfield2024,wu2024near,ma2022deep}, we explore the potential of sub-6G information for near-field beam selection to reduce pilot overhead.
Specifically, in order to account for the inherent uncertainty associated with the mapping from sub-6G information to mmWave beam selection, we propose to operate as follows: 
\begin{enumerate}[1)]
    \item \textbf{Sub-6G-based candidate beam selection}: Construct a candidate beam subset $\mathcal{C}(\widehat{\underline{\mathbf{H}}})$ based on the sub-6G channel estimate $\widehat{\underline{\mathbf{H}}}$.
    \item \textbf{MmWave beam training}: Perform limited mmWave beam training within the subset $\mathcal{C}(\widehat{\underline{\mathbf{H}}})$ to choose a beam $\mathbf{f}\in \mathcal{C}(\widehat{\underline{\mathbf{H}}})$.
\end{enumerate}

\par Traditional designs for candidate beam selection, such as top-$K$ \cite{alrabeiah2020deep} and probability sum-based methods \cite{ma2022deep} lack theoretical guarantees on the quality of the pre-selected candidate beam subset.
Thus, the subsequent mmWave beam training may fail to return a well-performing beam in the codebook $\mathcal{W}$ with a probability exceeding user’s requirements.

\par To formalize theoretical guarantees for the candidate beam subset, we first introduce a relaxed notion of beam optimality.
To this end, for any beam $\mathbf{f}$, we denote the \textit{suboptimality ratio} as the ratio of the rate achieved by beam $\mathbf{f}$ and the rate of the optimal beam $\mathbf{f}^{\star}$, i.e.,  
\begin{equation}
r({\mathbf f},\mathbf{H})=\frac{R(\mathbf{f},\mathbf{H})}{R(\mathbf{f}^{\star},\mathbf{H})}. 
\end{equation}  
Then, a beam $\mathbf{f}$ is said to be \textit{$\epsilon$-suboptimal} if it satisfies the condition 
\begin{eqnarray}
    r({\mathbf f},\mathbf{H})\geqslant 1-\epsilon,\label{eq:suboptimality condition}
\end{eqnarray}
where $\epsilon$, with $0\leqslant \epsilon \leqslant 1$, is the suboptimality factor.
In practice, the parameter $\epsilon$ is typically set in the range $0\leqslant \epsilon\leqslant 0.2$, with $\epsilon=0$ identifying the optimal beam $\mathbf{f}^{\star}$. 

\par In this paper, unlike prior studies \cite{cuiChannelEstimationExtremely2022,zhangFastNearFieldBeam2022,pengChannelEstimationExtremely2023}, we aim to design a sub-6G information-aided solution that guarantees a user-specified probability $1-\alpha \in [0,1]$ of identifying an $\epsilon$-suboptimal beam.
Specifically, given a target suboptimal ratio $\epsilon$, we wish to ensure that the probability of failing to obtain an $\epsilon$-suboptimal beam is no larger than $\alpha$, i.e., 
\begin{equation}
    \mathbf{Pr}\left(r(\mathbf{f},\mathbf{H})<1-\epsilon\right)\leqslant \alpha, \label{eq:suboptimal_prob}
\end{equation}
where the probability is taken over mmWave channels $\mathbf{H}$.
Condition \eqref{eq:suboptimal_prob} is satisfied, for a well-designed mmWave training phase, as long as the probability that the candidate beam subset $\mathcal{C}(\widehat{\underline{\mathbf{H}}})$ contains none beams with suboptimality smaller than $\epsilon$ does not exceed $\alpha$, i.e.,
\begin{equation}
\mathbf{Pr}\left(\nexists\mathbf{w} \in \mathcal{C}(\widehat{\underline{\mathbf{H}}}): r({\mathbf{w}},\mathbf{H}) \geqslant 1-\epsilon\right) \leqslant \alpha. \label{target coverage rate}
\end{equation}
We will refer to the probability in \eqref{target coverage rate} as the \textit{miscoverage probability}, and denote $1-\alpha$ as \textit{target coverage rate}.

\section{Sub-6G Channel Aided Near-field Beam Selection Framework}
\label{sec:method}
In this section, we first present an overview of SCAN-BEST, and then elaborate on its implementation.

\subsection{Overview of SCAN-BEST}
\begin{figure*}[t]
\centering
\includegraphics[width=0.96\textwidth]{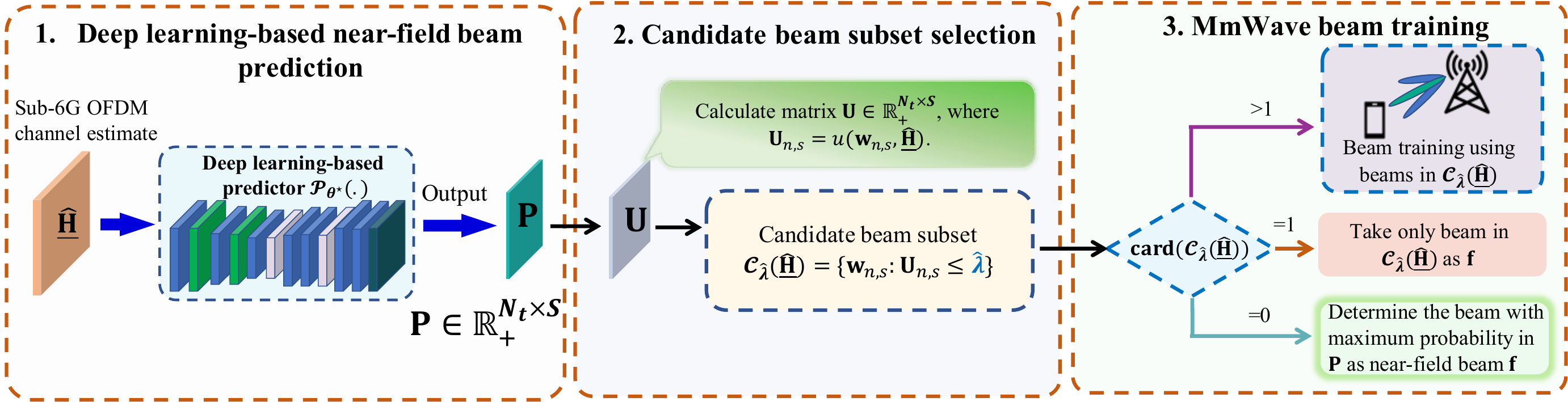}
\caption{The overall design of the SCAN-BEST framework.}
\label{SCAN-BEST diagram}
\end{figure*}
SCAN-BEST wraps around a deep learning-based beam predictor and on the statistical methodology of CRC \cite{angelopoulos2022conformal}.
For given suboptimaling parameter $\epsilon$ and probability $\alpha$, SCAN-BEST, as illustrated in Fig.~\ref{SCAN-BEST diagram}, proceeds along the following three stages:
\begin{enumerate}
    \item \textcolor{black}{\textbf{Deep learning-based near-field beam prediction:} Based on sub-6G channel estimate, employ any pre-trained neural model to assign each beam in the codebook $\mathcal{W}$ an estimated probability of being the optimal near-field beam \cite{dengwcnc2024}.}
    \item \textbf{Sub-6G-based candidate beam subset selection:} Using the predicted probabilities for all candidate beams in $\mathcal{W}$, CRC is leveraged to select a near-field candidate beam subset $\mathcal{C}(\widehat{\underline{\mathbf{H}}})$ that ensures the desired target coverage rate condition in \eqref{target coverage rate}.
    \item \textbf{MmWave beam training:} The near-field beam is determined by performing limited beam training only along the beam within the candidate beam subset $\mathcal{C}(\widehat{\underline{\mathbf{H}}})$.
\end{enumerate}

\subsection{Deep Learning-Based Near-Field Beam Prediction}
In the first stage, a deep learning model is used to predict the probability of each beam in the codebook $\mathcal{W}$ being the optimal near-field beam based sub-6G channel estimate $\widehat{\underline{\mathbf{H}}}$.
To elaborate, introduce a probability matrix ${\mathbf P}\in [0,1]^{ N_{\text{t}}\times S}$, where the entry ${\mathbf P}_{n,s}$ represents the predicted probability that the beam ${\mathbf w}_{n,s}$ is the optimal near-field beam.
The mapping between the sub-6G channel estimate $\widehat{\underline{\mathbf{H}}}$ and matrix ${\mathbf P}$, denoted as $\mathbf{P}=\mathscr{P}_{\boldsymbol{\theta}}(\widehat{\underline{\mathbf{H}}})$, is parameterized with a vector $\boldsymbol{\theta}$.

\par In order to train the neural network mapping $\mathbf{P}= \mathscr{P}_{\boldsymbol{\theta}}(\widehat{\underline{\mathbf{H}}})$, as in \cite{dengwcnc2024}, one assumes the availability of a dataset of $N_{\text{tr}}$ training samples, denoted as $\mathcal{D}_{\text{tr}}\triangleq\{\widehat{\underline{\mathbf{H}}}_i,n_i^\star,s_i^\star\}_{i=1}^{N_{\text{tr}}}$, associating a sub-6G channel estimate $\widehat{\underline{\mathbf{H}}}_i$ with the corresponding index $(n_i^\star,s_i^\star)$ of the optimal beam in \eqref{eq: optimal beam}.
Training is typically done by minimizing the standard cross-entropy loss, yielding the optimized parameter \cite{dengwcnc2024}
\begin{equation}
    \boldsymbol{\theta}^\star=\mathop{\arg\max}_{\boldsymbol{\theta}}\left\{-\sum\limits_{i=1}^{N_{\text{tr}}}\log \mathscr{P}_{\boldsymbol{\theta},n_i^{\star},s_i^{\star}}(\widehat{\underline{\mathbf{H}}})\right\}, \label{eq:trained_model}
\end{equation}
where $\mathscr{P}_{\boldsymbol{\theta},n_i^{\star},s_i^{\star}}(\widehat{\underline{\mathbf{H}}})$ is the $(n_i^{\star},s_i^{\star})$-th entry of matrix $\mathscr{P}_{\boldsymbol{\theta}}(\widehat{\underline{\mathbf{H}}})$.

\subsection{Sub-6G-Based Candidate Beam Subset Selection}
Given the trained model $\mathscr{P}_{\boldsymbol{\theta}^{\star}}(\cdot)$ in \eqref{eq:trained_model}, and given an input $\widehat{\underline{\mathbf{H}}}$, SCAN-BEST evaluates a negatively oriented score for beam $\mathbf{w}_{n,s}$, namely
\begin{equation}
u(n,s,\widehat{\underline{\mathbf{H}}}) = -\log P_{\text{max}}(\widehat{\underline{\mathbf{H}}})-\log {\mathbf P}_{n,s}(\widehat{\underline{\mathbf{H}}}), \label{eq:unlikelihood}
\end{equation}
where 
\begin{align}
    \mathbf{P}_{n,s}(\widehat{\underline{\mathbf{H}}})&=\mathscr{P}_{\boldsymbol{\theta}^{\star},n,s}(\widehat{\underline{\mathbf{H}}}),\notag\\
    P_{\text{max}}(\widehat{\underline{\mathbf{H}}})&=\max\limits_{n,s\in N_{\text{t}}\times S}\mathbf{P}_{n,s}(\widehat{\underline{\mathbf{H}}}).
\end{align}
The term $-\log P_{\text{max}}(\widehat{\underline{\mathbf{H}}})$ is referred to as the R\'enyi min-entropy for the predictive conditional probability $\mathscr{P}_{\boldsymbol{\theta}^{\star}}(\widehat{\underline{\mathbf{H}}})$, which provides a measure of the uncertainty of the model $\mathscr{P}_{\boldsymbol{\theta}^{\star}}(\cdot)$ regarding this prediction with input $\widehat{\underline{\mathbf{H}}}$ \cite{simeone_cqit}.
Adding this term to the negative log-likelihood $-\log \mathbf{P}_{n,s}(\widehat{\underline{\mathbf{H}}})$ degrades the score $u(n,s,\widehat{\underline{\mathbf{H}}})$ when the predictive uncertainty is high.
Recall, in fact, that a smaller value of the score $u(n,s,\widehat{\underline{\mathbf{H}}})$ indicates that the beam $\mathbf{w}_{n,s}$ is predicted to be more likely to be optimal.

\par The subset $\mathcal{C}_{\lambda}(\widehat{\underline{\mathbf{H}}})$ of candidate beams includes all beams in set $\mathcal{W}$ whose scores \eqref{eq:unlikelihood} are no larger than a threshold $\lambda$, i.e.,
\begin{equation}
\mathcal{C}_{\lambda}(\widehat{\underline{\mathbf{H}}})=\left\{\mathbf{w}_{n,s} \in \mathcal{W}:u(n,s,\widehat{\underline{\mathbf{H}}})\leqslant \lambda\right\}. \label{def candidate beam subset}
\end{equation}

\par In order to select the threshold $\lambda$ so that the coverage condition \eqref{target coverage rate} is satisfied, we adopt CRC \cite{angelopoulos2022conformal}.
To this end, we assume a held-out calibration dataset consisting of $N_{\text{cal}}$ calibration samples, defined as $\mathcal{D}_{\text{cal}} \triangleq \{ (\widehat{\underline{\mathbf{H}}}_i, \mathbf{H}_i) \}_{i=1}^{N_{\text{cal}}}$, where $\widehat{\underline{\mathbf{H}}}_i$ denotes the sub-6~GHz channel estimate and $\mathbf{H}_i$ is the corresponding ground-truth mmWave channel.
Using the calibration dataset, the probability \eqref{target coverage rate} is estimated, and the estimate is evaluated as a function of the threshold $\lambda$.
This evaluation is leveraged to find a threshold that satisfies the inequality \eqref{target coverage rate}.

\par Specifically, the set $\mathcal{C}_{\lambda}(\widehat{\underline{\mathbf{H}}}_i)$ is evaluated using \eqref{def candidate beam subset} for all calibration data points $i=1,\dots,N_{\text{cal}}$.
Then, the miscoverage probability \eqref{target coverage rate} is estimated using the calibration data set as 
\begin{equation}
    \hat{R}(\lambda)=\frac{1}{N_{\text{cal}}}\sum\limits_{i=1}^{N_{\text{cal}}}\mathds{1}\left(\nexists \mathbf{w}\in \mathcal{C}_{\lambda}(\widehat{\underline{\mathbf{H}}}_i): r(\mathbf{w},\mathbf{H}_i)\geqslant 1-\epsilon\right), \label{eq: empirical risk function}
\end{equation}
where $\mathds{1}(\cdot)$ is a conventional indicator function, which returns 1 when the event is true and 0 otherwise.
\begin{proposition}
    \label{prop: empirical risk function}
    The subset $\mathcal{C}_{\hat{\lambda}}(\widehat{\underline{\mathbf{H}}})$ in \eqref{def candidate beam subset} with the threshold
    \begin{equation}
    \hat{\lambda}=\inf\left\{\lambda\in \mathbb{R}:\hat{R}(\lambda)\leqslant\alpha+\frac{\alpha-1}{N_{\text{cal}}}\right\}\label{eq: hat lambda optimization}
    \end{equation} 
    attains the target coverage rate $1-\alpha$ in \eqref{target coverage rate} for any input $\widehat{\underline{\mathbf{H}}}$ and ``any'' predictive model $\mathscr{P}_{\boldsymbol{\theta}^{\star}}(\cdot)$.
\end{proposition}
\IEEEproof The proof is deferred to Appendix~\ref{appendix:A}.
\begin{figure}[t]
\centering
\includegraphics[width=0.42\textwidth]{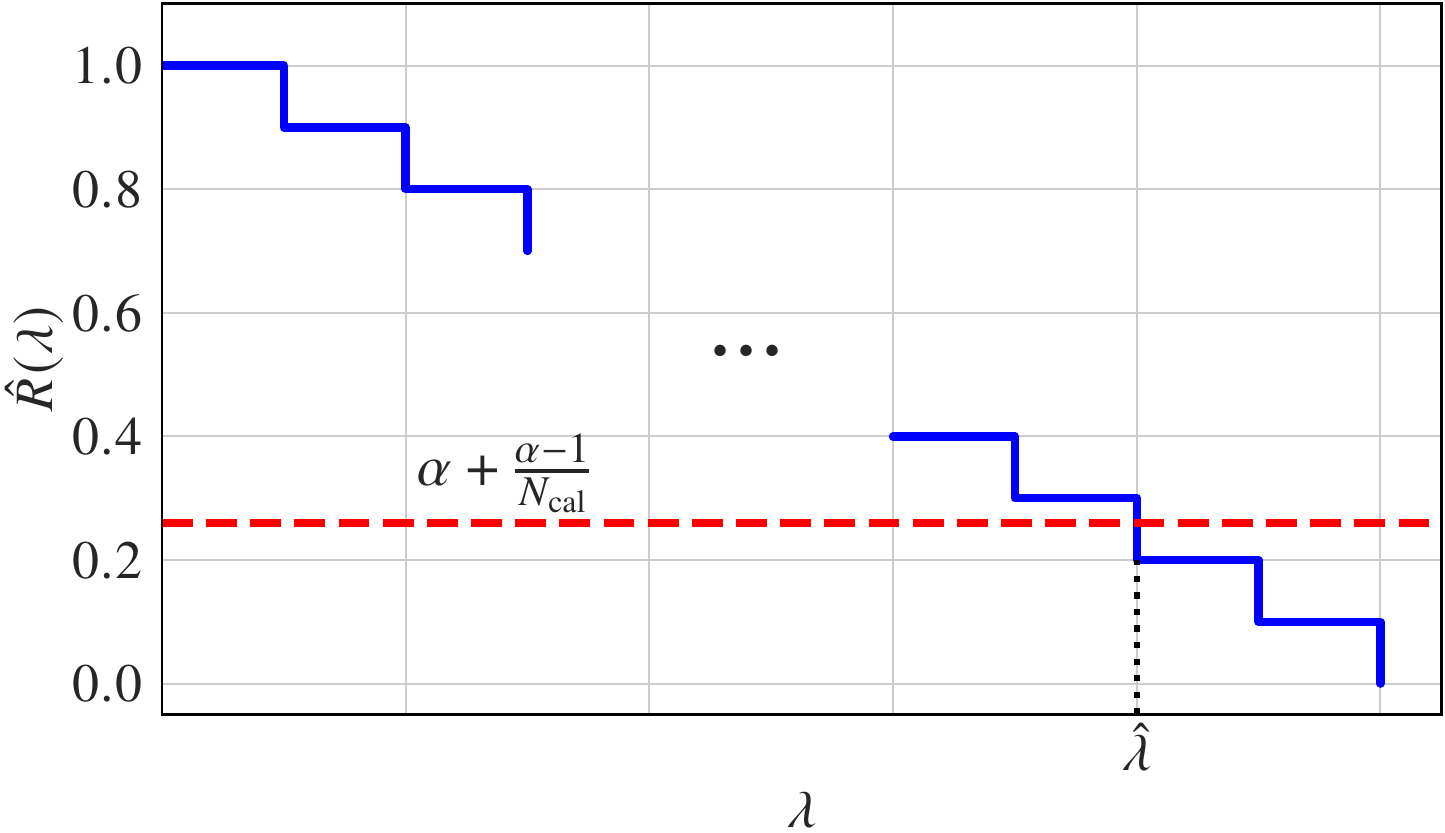}
\caption{A simple illustration of $\hat{R}(\lambda)$ (blue line), the orange line denotes the desired risk level.}
\label{illustrated_fig_R_lambda}
\end{figure}

\par Accordingly, as illustrated in Fig.~\ref{illustrated_fig_R_lambda}, the threshold is selected so that the miscoverage probability \eqref{eq: empirical risk function} satisfies the condition \eqref{target coverage rate} with the caveat that the reliability requirement is made stricter\textemdash from $\alpha$ to $\alpha+(\alpha-1)/N_{\text{cal}}$.
Note that the added term in the coverage probability $\alpha$ decreases to zero as $N_{\text{cal}}\rightarrow \infty$.

\par A threshold satisfying the condition \eqref{eq: hat lambda optimization} always exists, since the choice $\lambda=\max_{i\in \{1,\dots,N_{\text{cal}}\}} \allowbreak u(n_i^{\star},s_i^{\star},\widehat{\underline{\mathbf{H}}}_i)$ yields $\hat{R}(\lambda)=0$.
While the threshold in~\eqref{eq: hat lambda optimization} can be obtained via exhaustive search over all possible values of $\lambda$ or through hierarchical methods (e.g., binary search), these approaches typically involve relatively high computational complexity. 
To address this, we propose a low-complexity algorithm for efficiently determining the threshold in~\eqref{eq: hat lambda optimization}, as detailed in Appendix~\ref{appendix: hat lambda}.

\subsection{Addressing Covariate Shifts between Calibration and Test Data}
The theoretical guarantees of CRC rely on the assumption that calibration and test data follow the same distribution.
In this subsection, we consider a more general case, in which a covariate shift exists between the calibration and test data. 
In this case, the probability distribution of sub-6G channel estimates $\widehat{\underline{\mathbf{H}}}$ in the calibration dataset, denoted as $\Gamma_{\text{cal}}(\widehat{\underline{\mathbf{H}}})$, differs from that in the test dataset, denoted by $\Gamma_{\text{te}}(\widehat{\underline{\mathbf{H}}})$.
To address this distribution mismatch, we adopt the weighted CRC method \cite{tibshirani2019conformal}.

\par Weighted CRC assigns importance weights to each sample in the calibration dataset to compensate for the discrepancy between calibration and test data.
Specifically, the importance weight for the $i$-th calibration sample and the importance weight for the test input $\widehat{\underline{\mathbf{H}}}'$, are defined as
\begin{equation}
\omega_i = \frac{p_i}{\sum_{j=1}^{N_{\text{cal}}}p_j+p'},~\text{and}~\omega'= \frac{p'}{\sum_{j=1}^{N_{\text{cal}}}p_j+p'},\label{eq:importance weight}
\end{equation}
respectively, where the likelihood ratio $p_i$ for $i=1,\dots,N_{\text{cal}}$ and $p'$ are given by
\begin{equation}
p_i=\frac{\Gamma_{\text{te}}(\widehat{\underline{\mathbf{H}}}_i)}{\Gamma_{\text{cal}}(\widehat{\underline{\mathbf{H}}}_i)},~\text{and}~p'= \frac{\Gamma_{\text{te}}(\widehat{\underline{\mathbf{H}}}')}{\Gamma_{\text{cal}}(\widehat{\underline{\mathbf{H}}}')}.
\end{equation}

\par With the weights \eqref{eq:importance weight}, the weighted miscoverage probability \eqref{eq: empirical risk function} is defined as 
\begin{equation}
    \tilde{R}(\lambda) =  \sum\limits_{i=1}^{N_{\text{cal}}} 
    \omega_i\mathds{1} \left( 
        \nexists \mathbf{w} \in \mathcal{C}_{\lambda}(\widehat{\underline{\mathbf{H}}}_i) : 
        r(\mathbf{w}, \mathbf{H}_i) \geqslant 1 - \epsilon 
    \right).
\end{equation}
Intuitively, this estimate assigns more importance to the samples in the calibration dataset that are more likely to be samples from the test distribution.
\begin{proposition}
    \label{prop: weighted_CRC}
The subset $\mathcal{C}_{\hat{\lambda}}(\widehat{\underline{\mathbf{H}}})$ in \eqref{def candidate beam subset} with the threshold 
\begin{equation}
    \hat{\lambda} = \inf \left\{ \lambda \in \mathbb{R} : \tilde{R}(\lambda) \leqslant \alpha - \omega' \right\}\label{eq: weighted CRC threshold}
\end{equation}
attains the target coverage rate $1-\alpha$ in the test data.
\end{proposition}
\IEEEproof The proof is deferred to Appendix~\ref{appendix:B}. 

\par In practical implementations, the true densities $\Gamma_{\text{cal}}(\cdot)$ and $\Gamma_{\text{te}}(\cdot)$ are unknown. 
Therefore, following \cite{tibshirani2019conformal}, we estimate the importance weights using a pre-trained probabilistic classifier $g(\cdot)$ optimized to distinguish between samples from calibration and test distribution.
Accordingly, the model $g(\cdot)$ outputs the probability that the test input $\widehat{\underline{\mathbf{H}}}'$ belongs to the test distribution. 
The estimated likelihood ratios $\hat{p}_i$ for $i=1,\dots,N_{\text{cal}}$ and $\hat{p}'$ are then computed as
\begin{equation}
    \hat{p}_i = \frac{g(\widehat{\underline{\mathbf{H}}}_i)}{1 - g(\widehat{\underline{\mathbf{H}}}_i)},~\text{and}~\hat{p}'= \frac{g(\widehat{\underline{\mathbf{H}}}')}{1 - g(\widehat{\underline{\mathbf{H}}}')}.
\end{equation}

\subsection{MmWave Beam Training}
For a given input $\widehat{\underline{\mathbf{H}}}$, the candidate beam subset $\mathcal{C}_{\hat{\lambda}}(\widehat{\underline{\mathbf{H}}})$ may have any cardinality between zero and $N_{\text{t}} S$.
Accordingly, the mmWave training stage operates as follows:
\begin{enumerate}[1)]
    \item If the set ${\mathcal C}_{\hat{\lambda}}(\widehat{\underline{\mathbf{H}}})$ is empty, SCAN-BEST selects the most-likely beam corresponding to the maximum probability in $\mathbf{P}$ as the final beam $\mathbf{f}$.
    \item If the set ${\mathcal C}_{\hat{\lambda}}(\widehat{\underline{\mathbf{H}}})$ has cardinality equal to one, the only beam in ${\mathcal C}_{\hat{\lambda}}(\widehat{\underline{\mathbf{H}}})$ is selected as the final beam $\mathbf{f}$.
    \item If the set ${\mathcal C}_{\hat{\lambda}}(\widehat{\underline{\mathbf{H}}})$ has cardinality larger than one, mmWave beam training is performed to select final beam $\mathbf{f}$.    
\end{enumerate}

\par MmWave training leverages the transmission of pilot symbols in the uplink of the mmWave band.
To elaborate, let $s_{\text{p},m}=\sqrt{P_{\text{s}}/M} $ denote the uplink pilot signal at the $m$-th subcarrier, where \( P_{\text{s}} \) represents the total power allocated for pilot transmission.  
When employing the beam $\mathbf{w}\in \mathcal{C}_{\hat{\lambda}}(\widehat{\underline{\mathbf{H}}})$, the corresponding received signal at the $m$-th subcarrier, denoted as \( y_{\mathbf{w},m} \), is given by:
\begin{equation}
y_{\mathbf{w},m}={\mathbf w}^T \mathbf{h}_m~s_{\text{p},m}+ {\mathbf w}^T \mathbf{n}_m, \label{eq:y_w_k}
\end{equation}
where $\mathbf{n}_m\sim\mathcal{CN}(\mathbf{0},\sigma_{\text{n}}^2\mathbf{I}_{N_{\text{t}}})$ denotes the noise vector with variance $\sigma_{\text{n}}^2$.
Based on the received signal \eqref{eq:y_w_k}, the final beam $\mathbf{f}$ is selected as the beam corresponding to the maximum received power:
\begin{equation}
{\mathbf f}=\mathop{\arg\max}\limits_{{\mathbf w}\in {\mathcal C}_{\hat{\lambda}}(\widehat{\underline{\mathbf{H}}})} \sum_{m=1}^{M}|y_{\mathbf{w},m}|^2. \label{eq:final_beam_selection}
\end{equation}
If multiple pilots are transmitted per beam, the average power is computed in \eqref{eq:final_beam_selection}.

\par The entire implementation of SCAN-BEST framework is summarized in Algorithm \ref{SCAN-BEST algorithm}.
\begin{algorithm}[!h]
    \small
    \renewcommand{\algorithmicrequire}{\textbf{Input:}}
    \renewcommand{\algorithmicensure}{\textbf{Output:}}
    \caption{SCAN-BEST}
    \label{SCAN-BEST algorithm}
    \begin{algorithmic}[1]
        \Require Sub-6G OFDM channel estimate $\underline{\widehat{\mathbf H}}$, target miscoverage rate $\alpha$ 
        \Ensure Near-field beam ${\mathbf f}$ 
        \State Produce the near-field beam probability matrix $\mathscr{P}_{\boldsymbol{\theta}^{\star}}(\widehat{\underline{\mathbf{H}}})$
        \State Construct near-field candidate beam subset ${\mathcal C}_{\hat{\lambda}}(\widehat{\underline{\mathbf{H}}})$ via \eqref{def candidate beam subset}
        \If{${\mathcal C}_{\hat{\lambda}}(\widehat{\underline{\mathbf{H}}})$ is empty, i.e., $\text{card}({\mathcal C}_{\hat{\lambda}}(\widehat{\underline{\mathbf{H}}}))=0$}
        \State Identify the beam with the largest probability in $\mathbf{P}=\mathscr{P}_{\boldsymbol{\theta}^{\star}}(\widehat{\underline{\mathbf{H}}})$ as the near-field beam $\mathbf{f}$
        \ElsIf{$\text{card}({\mathcal C}_{\hat{\lambda}}(\widehat{\underline{\mathbf{H}}}))=1$}
        \State Assign the only beam in $\mathcal{C}_{\hat{\lambda}}(\widehat{\underline{\mathbf{H}}})$ as $\mathbf{f}$
        \Else
        \State Conduct uplink beam training using the beams in ${\mathcal C}_{\hat{\lambda}}(\widehat{\underline{\mathbf{H}}})$ and then select the beam with the strongest received signal power as $\mathbf{f}$ using \eqref{eq:final_beam_selection}
        \EndIf
    \end{algorithmic}
\end{algorithm}

\section{Numerical Results}
In this section, we present experimental results to validate the performance of SCAN-BEST.

\subsection{Simulation Scenario and System Parameter Setup}
\label{sec:numerical results}
\begin{figure}[t]
\centering
\includegraphics[width=0.43\textwidth]{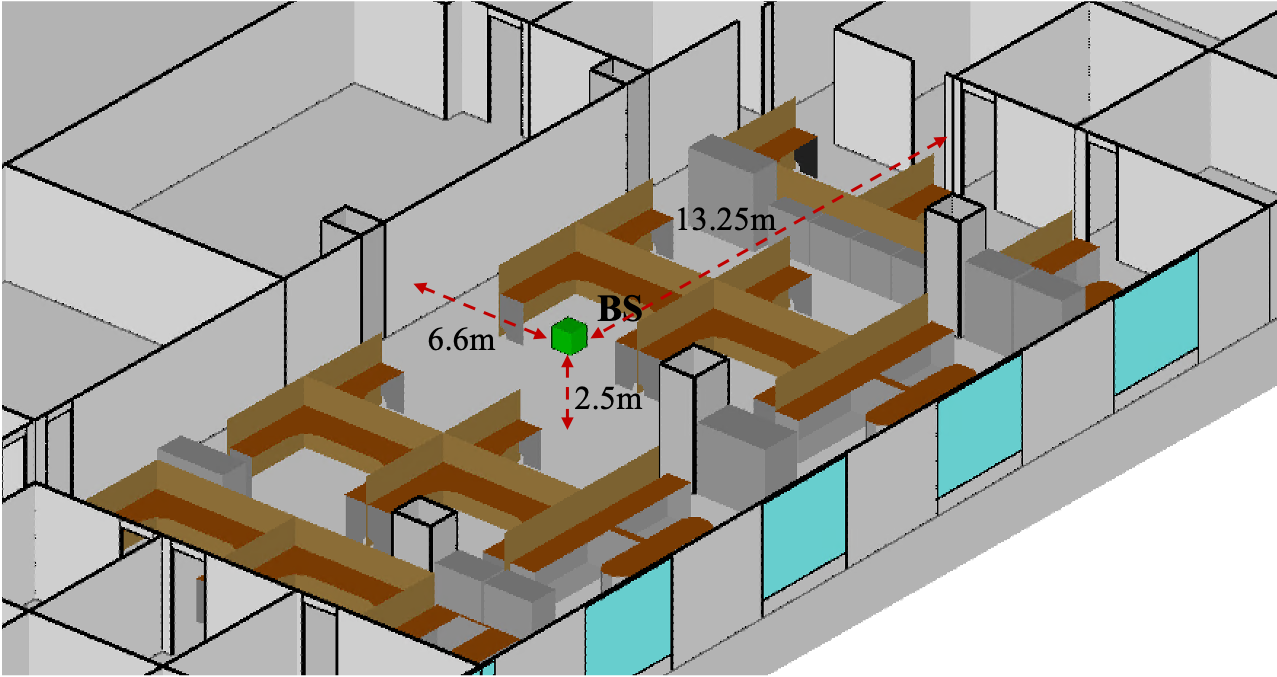}
\caption{An illustration of simulation setup.}
\label{indorr simulation region}
\end{figure}
\begin{table}[htbp]
\linespread{1.0}\selectfont
\caption{System Parameters}
\label{detailed parameters}
\centering
\setlength\tabcolsep{0.5mm}{
\begin{tabular}{c|c|c}
\toprule
Notations                                   & Parameters                       & Values    \\ 
\midrule
\midrule
$f_{\text{c}}$,~$\underline{f_{\text{c}}}$                    & \makecell{Operating frequency of \\ mmWave, sub-6G system  (GHz)}        & 73, 3.5 \\ \hline
$W$,~$\underline{W}$& \makecell{Bandwidth of \\mmWave, sub-6G system (MHz) }& 200, 80\\ \hline
$M$,~$\underline{M}$ & \makecell{Number of subcarriers \\ of mmWave, sub-6G system} & 64, 32  \\ \hline
$N_{\text{t}}$,~$\underline{N_{\text{t}}}$ & \makecell{Number of mmWave, sub-6G \\antennas at the BS} & 256, 16  \\ \hline
$P_{\text{t}}$,~$P_{\text{s}}$& \makecell{Total power of mmWave \\downlink, uplink pilot (dBm)}& 25,~25\\ \hline
$\underline{P_{\text{s}}}$ & \makecell{Power of sub-6G \\pilot signals (dBm) }& 10\\ \hline
$\sigma_{\text{n}}^2$,~$\underline{\sigma_{\text{n}}^2}$ & \makecell{Noise power of \\mmWave, sub-6G system (dBm) }& {\tiny\makecell{$-173.8 + 90 + 10\log_{10}(W)$\\$-173.8 + 90 + 10\log_{10}(\underline{W})$}}\\ \hline
$S$, ~$\beta$ & \makecell{Number of candidate distances\\and correlation parameter of $\mathcal{W}$}& 7,~1.6\\
\bottomrule
\end{tabular}}
\end{table}
As depicted in Fig.~\ref{indorr simulation region}, we consider an indoor scenario within an area of dimensions $13.2~\text{m} \times 26.5~\text{m}$.
The BS with height 2.5 m, is located at the center of the area.
The UE is assumed at to be the height of 1 m, and it can be located anywhere within the room, leading to the LoS condition or NLoS conditions.
More detailed system parameters are listed in Table~\ref{detailed parameters}, with exceptions marked explicitly in the text.

\par We first collect a total of 10,000 samples from the above indoor scenario via a ray-tracing software \cite{Remcom}, each of which consists of a pair of sub-6G channel estimate and true mmWave channel $\{\widehat{\underline{\mathbf H}},{\mathbf H}\}$.
These samples are randomly split for training, validating, calibration, and test, respectively, with the ratio of 50\%, 10\%, 20\%, and 20\%, respectively.

\subsection{Performance Metrics and Baselines}
\begin{table*}[t]
    \linespread{1.0}\selectfont
    \setlength{\tabcolsep}{3pt}
    \centering
    \caption{Baseline Methods and Our Method}
    \label{table:baseline_comparison}
    \begin{tabular}{@{}c c c c C{5cm}@{}}
        \toprule
        \textbf{Method Category} & \textbf{Method} & \textbf{Predictor} & \textbf{Candidate beam subset} & \multicolumn{1}{c}{\textbf{MmWave Pilot Usage}} \\
        \midrule
        \multirow{2}{*}{\makecell[c]{In-band mmWave\\training baselines}} 
            & EBS \cite{zhangFastNearFieldBeam2022} & \XSolidBrush & \XSolidBrush & Exhaustive search with $N_{\text{t}} S$ pilots \\
            & FNBS \cite{zhangFastNearFieldBeam2022}& \XSolidBrush & \XSolidBrush & Two-stage search with $N_{\text{t}} + 3S$ pilots \\
        \midrule
        \multirow{5}{*}{\makecell[c]{Sub 6GHz-based\\baselines}} 
            & \multirow{2}{*}{\textcolor{black}{CSLW-NBS} \cite{wu2024near}} & \multirow{2}{*}{\XSolidBrush} & \multirow{2}{*}{\XSolidBrush} & \textcolor{black}{MmWave pilot measurements for near-field channel estimation} \\
            \cmidrule{5-5}
            & SPBP + Top-$K$ & SPBP & Top-$K$ & \multirow{7}{*}{\makecell[c]{Limited mmWave beam training \\within candidate beam subset}} \\
            & SPBP + PS & SPBP & PS & \\
            & ADADT-P + Top-$K$ & ADADT-P & Top-$K$ & \\
            & ADADT-P + PS & ADADT-P & PS & \\
            \cmidrule{1-4}
            \multirow{2}{*}{\makecell[c]{\textbf{Our method}}}& \multirow{2}{*}{\makecell[c]{\textbf{SCAN-BEST}}} & \textcolor{black}{SPBP} & \multirow{2}{*}{\makecell[c]{CRC (optimally weighted \\for covariate shifts)}}& \\
            & & \textcolor{black}{ADADT-P} & & \\
        \bottomrule
    \end{tabular}
\end{table*}
To comprehensively assess the performance of proposed SCAN-BEST, we evaluate the following performance metrics:
\begin{itemize}
    \item \textbf{Achieved coverage rate}, i.e., the complement of the probability \eqref{target coverage rate}, quantifying the probability that a beam in $\mathcal{C}_{\hat{\lambda}}(\widehat{\underline{\mathbf{H}}})$ satisfies the condition in \eqref{eq:suboptimality condition} over the test dataset;
    \item \textbf{Achieved $\epsilon$-suboptimal rate}, i.e., the complement of the probability \eqref{eq:suboptimal_prob}, quantifying the probability that final selected beam $\mathbf{f}$ is $\epsilon$-suboptimal over the test dataset.
\end{itemize}

\par As for the predictor, we consider the following in the comparison:
\begin{itemize}
    \item \textbf{Augmented discrete angle and delay transformation-based prediction (ADADT-P):} ADADT-P begins by applying an augmented two-dimensional discrete Fourier transform \cite{sun2018single} to the sub-6G channel estimates to extract angle and delay information as key features. These features are then fed into a neural network to predict the optimal near-field beam probabilities.
    \item \textbf{Sub-6G pilots-based beam prediction (SPBP):} \textcolor{black}{Adapted from \cite{liuNMBEnetEfficientNearfield2024}, SPBP uses a neural network to directly predict the optimal near-field beam probabilities from received sub-6G pilot signals.}
\end{itemize}
Both the detailed configurations and training settings of the above predictors are provided in Appendix~\ref{appendix:C}.
\begin{figure}[htbp]
\centering 
\includegraphics[width=0.43\textwidth]{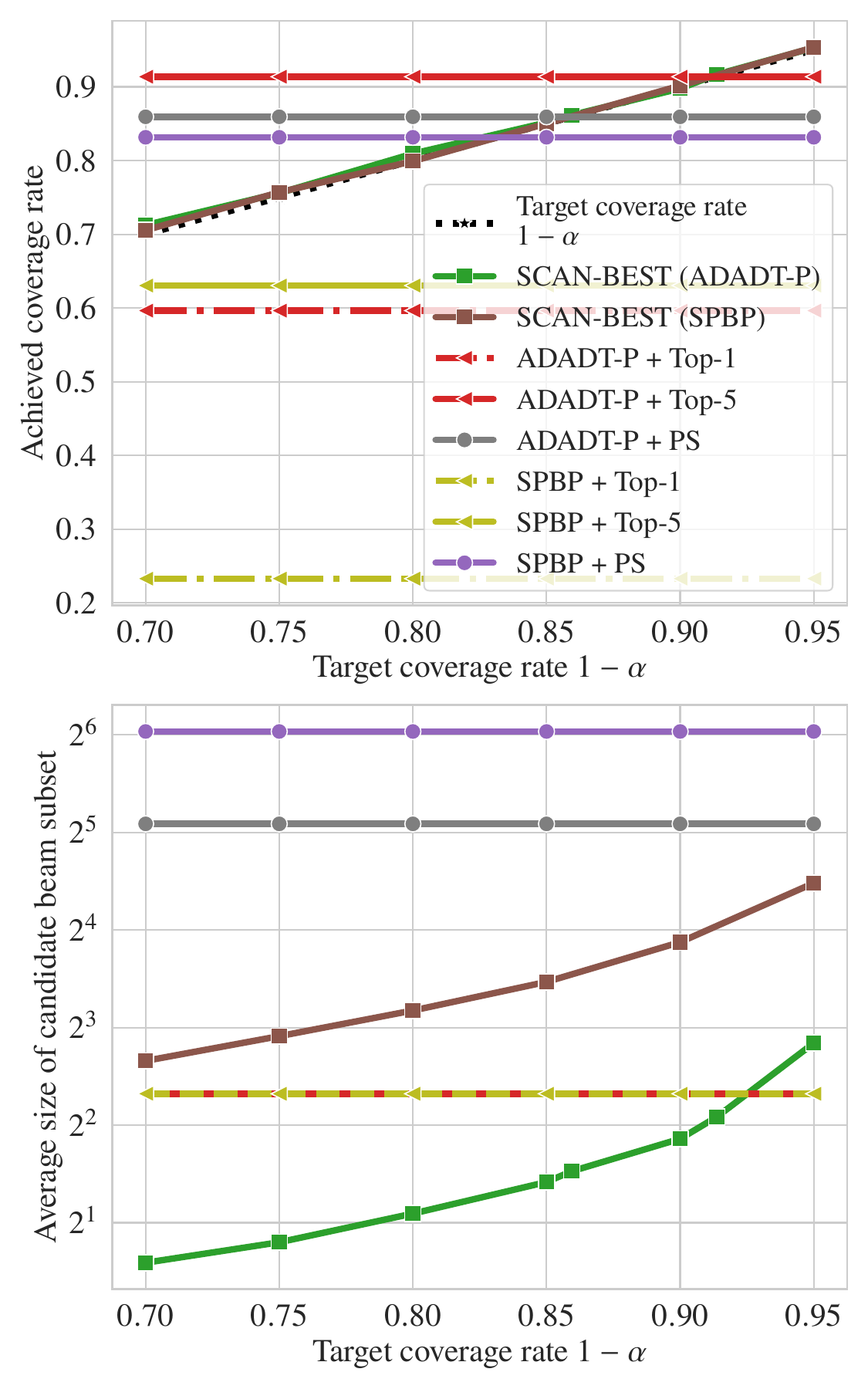}
\caption{Achieved coverage rates and average sizes of candidate beam subset of SCAN-BEST and baselines for $\epsilon=0.15$.}
\label{fig:coverage_and_size}
\end{figure}
\par As benchmarks, we consider two classical schemes to select the candidate beam subset based on the given predictor above:
\begin{itemize}
    \item \textbf{Top-$K$ selection \cite{alrabeiah2020deep}:} Identify the $K$ beams with the highest predictive probabilities to form the candidate beam subset. 
    \item \textbf{Probability sum (PS) selection \cite{ma2022deep}:} Choose the smallest set of beams whose cumulative predictive probability exceeds a predefined threshold (set to 0.99 unless otherwise specified) to form the candidate beam subset.
\end{itemize}
\par In Table~\ref{table:baseline_comparison}, we summarize different combinations of predictor and candidate beam subset selection scheme.
Furthermore, we include an additional baseline referred to as CSLW-NBS, which is adapted from the recently proposed sub-6 GHz-aided near-field channel estimation method, CSLW-BOMP, introduced in \cite{wu2024near}. 
Specifically, CSLW-NBS selects the beam from the near-field codebook that has the highest correlation with the near-field channel estimated by the CSLW-BOMP algorithm.
In addition, we consider two baseline methods that do not utilize sub-6 GHz information: Exhaustive Beam Search (EBS) and Fast Near-Field Beam Search (FNBS), both originally proposed in \cite{zhangFastNearFieldBeam2022}.

\subsection{Coverage Rate Guarantee of SCAN-BEST}
\begin{figure}[htbp]
    \centering 
    \includegraphics[width=0.45\textwidth]{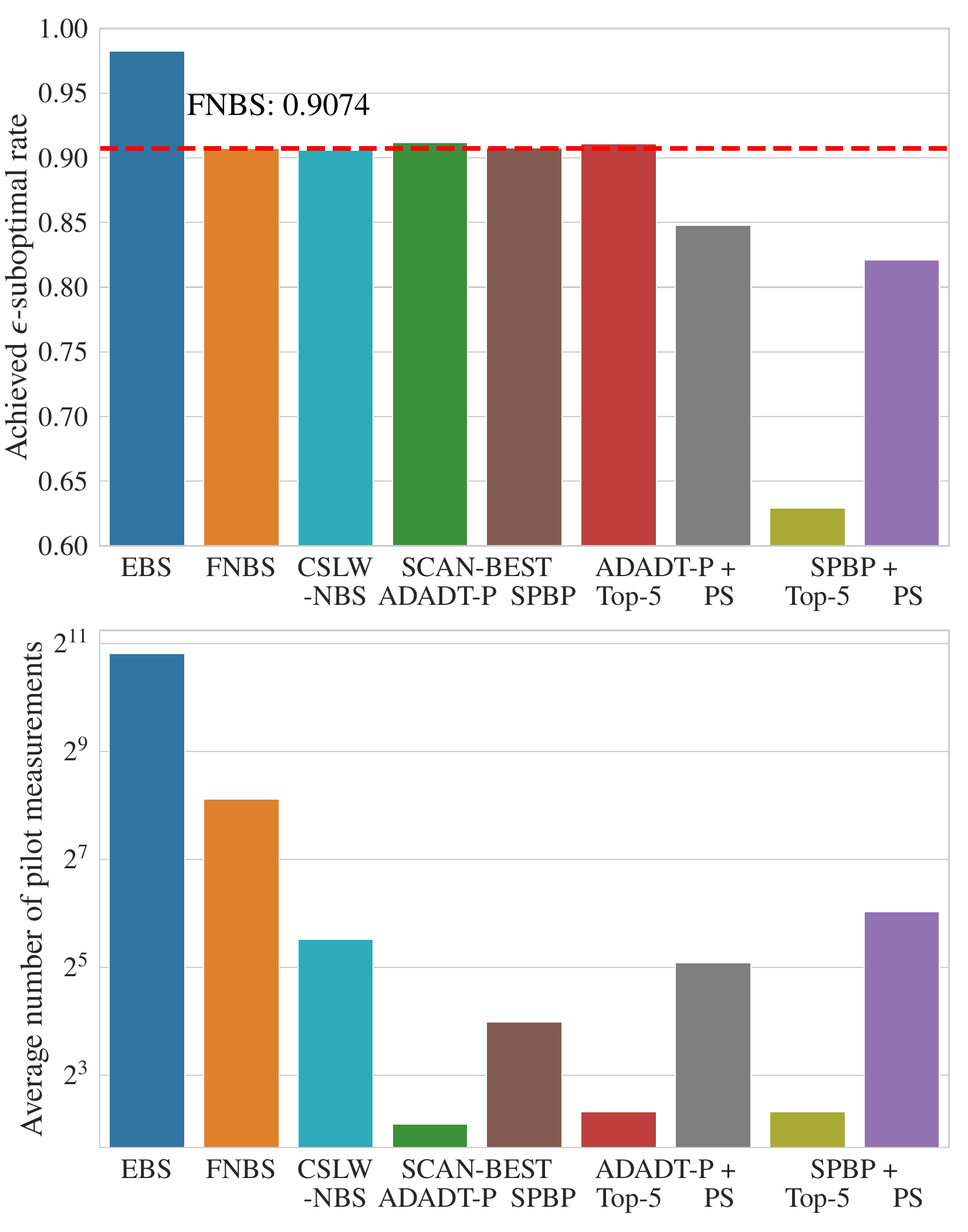}
    \caption{Comparison of the average suboptimality ratio and of the number of pilot symbols for SACN-BEST, mmWave-based, and sub 6GHz-based baselines.}
    \label{fig:suboptimality_ratio_and_num_pilot}
    \end{figure}
To start, we evaluate the reliability guarantees provided by SCAN-BEST in term of coverage rate \eqref{target coverage rate}.
Fig.~\ref{fig:coverage_and_size} presents the achieved coverage rates, along with the sizes of the candidate beam subsets for SCAN-BEST and for sub 6GHz-based baselines when the suboptimality target is $\epsilon=0.15$.
For results with other values of $\epsilon$, refer to Fig.~\ref{fig:coverage_and_size:part2} in Appendix~\ref{appendix:C}.
Confirming theoretical properties of Proposition \ref{prop: empirical risk function}, with the help of CRC, both the predictors ADADT-P and SPBP can be calibrated to the target coverage rate by dynamically increasing the size of candidate beam subset.
Thanks to the more use of a more effective predictor, SCAN-BEST (ADADT-P) requires a smaller candidate beam subset size as compared to SCAN-BEST (SPBP) to achieve the same coverage rate.

\par In contrast, without applying CRC and directly selecting the codeword with the highest predicted probability as the beam, both ADADT-P + Top-1 and SPBP + Top-1 achieve only fixed and relatively low coverage rates. 
When considering using alternative candidate beam subset selection schemes - Top-$K$ ($K>1$) or PS, the achieved coverage rates of corresponding baselines is indeed higher, but remain fixed and cannot adapt to varying target coverage rates.
Furthermore, when achieving the same coverage rate as ADADT-P + Top-5 and ADADT-P + PS, SCAN-BEST (ADADT-P) requires smaller candidate beam subsets. 
Specifically, the average sizes of the candidate beam subsets for ADADT-P + Top-5 and ADADT-P + PS=0.99 are 5 and 33.9 for $\alpha=0.086$ and $\alpha=0.106$, respectively, while those for SCAN-BEST (ADADT-P) are 4.2 and 2.8, respectively.
Note that the superiority of SCAN-BEST (ADADT-P) over SCAN-BEST (SPBP) stems from the fact that the employed ADADT explicitly extracts critical angle and delay features from sub-6~GHz channel estimates, which facilitates easier training and leads to improved model performance.
\begin{figure*}[!t]
\centering 
\subfloat[$1-\alpha=0.75$]{\includegraphics[width=0.33\textwidth]{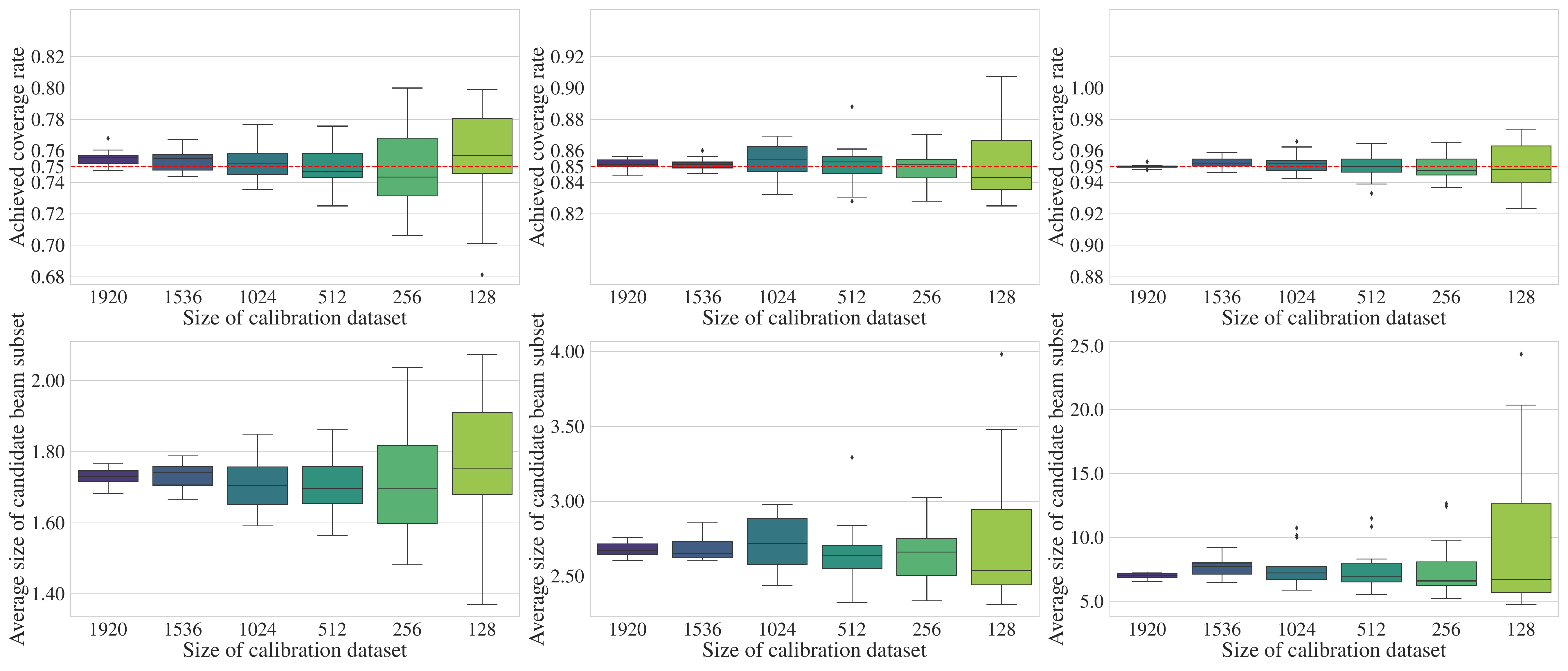}}\label{fig:size_of_calibration_dataset:75}
\subfloat[$1-\alpha=0.85$]{\includegraphics[width=0.33\textwidth]{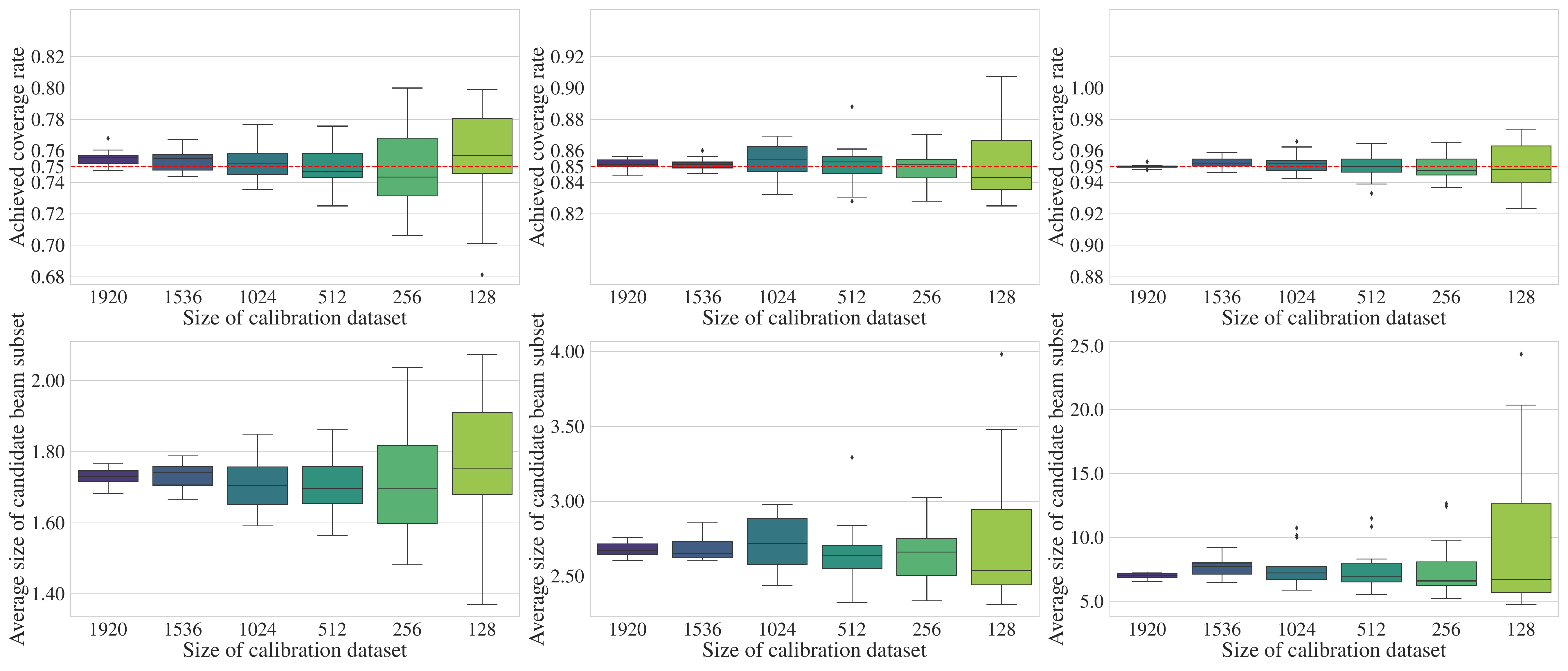}\label{fig:size_of_calibration_dataset:85}}
\subfloat[$1-\alpha=0.95$]{\includegraphics[width=0.33\textwidth]{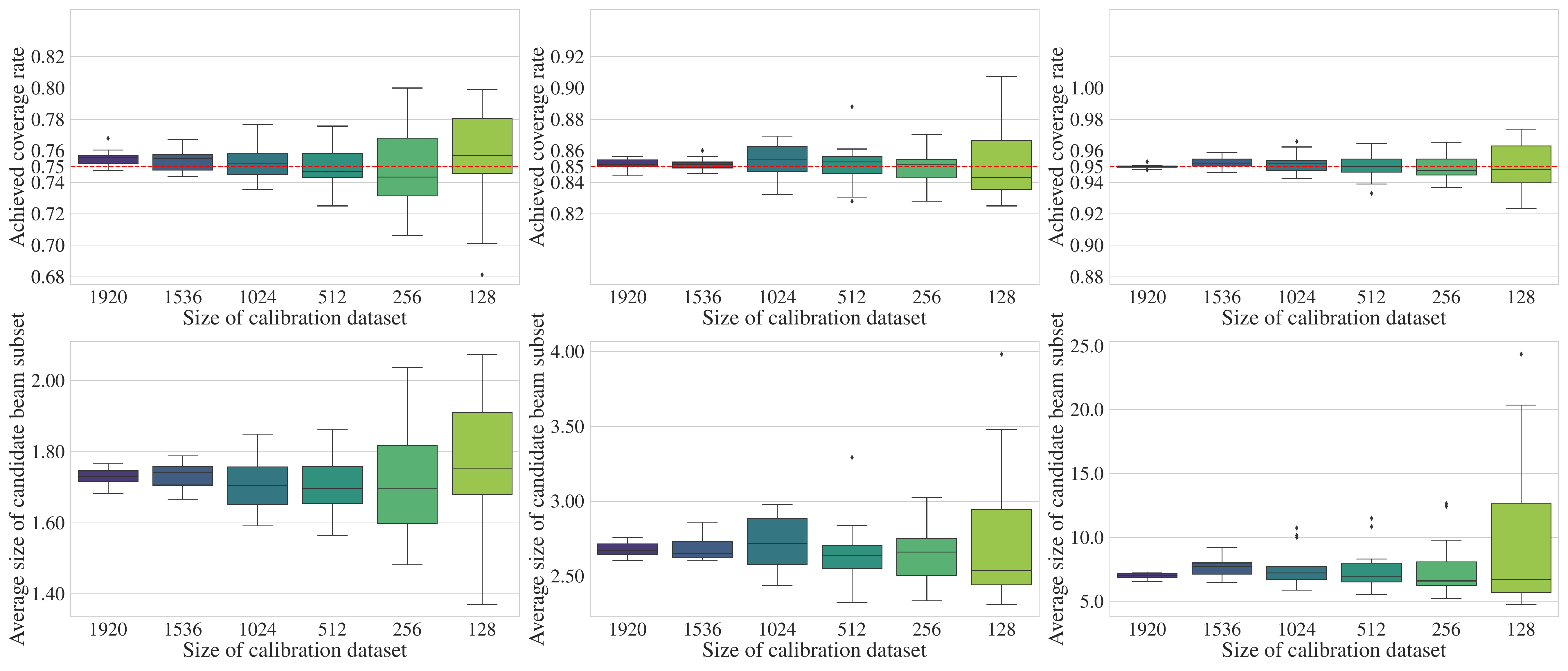}}\label{fig:size_of_calibration_dataset:95}
\caption{Box plots for the achieved coverage rate and the size of the candidate beam subset under different calibration dataset sizes, with target coverage rates $1-\alpha$ of (a) 0.75, (b) 0.85, and (c) 0.95. The red lines denote the target coverage rates $1-\alpha$, while ``M'' indicates the mean.}
\label{fig:size_of_calibration_dataset}
\end{figure*}

\par In addition, we implement a comparison among in-band mmWave training baselines, CSLW-NBS, and sub-6GHz-based deep learning based baselines, as shown in Fig.~\ref{fig:suboptimality_ratio_and_num_pilot}.
EBS achieves the best performance but requires an impractical 1792 mmWave pilots, while FNBS attains an $\epsilon$-suboptimal rate of 0.9074 using only 277 mmWave pilots.
Compared with the in-band mmWave training baselines, both CSLW-NBS and SCAN-BEST offer a more favorable trade-off between performance and pilot overhead. 
Specifically, CSLW-NBS achieves performance comparable to FNBS with only 46 mmWave pilots.
Furthermore, under a target coverage rate of $1 - \alpha = 0.91$, SCAN-BEST (ADADT-P) and SCAN-BEST (SPBS) achieve higher $\epsilon$-suboptimal rates than FNBS, while requiring only an average of 4.2 and 15.8 mmWave pilots, respectively.
These results highlight the superior performance of SCAN-BEST over both CSLW-NBS and in-band training baselines, demonstrating the effectiveness of deep learning techniques in feature extraction and predictive modeling.

\subsection{Impact of the Size of Calibration Dataset}
Here, we analyze the achieved coverage rate and candidate beam subset size when varying the calibration dataset sizes.
Fixing the suboptimality parameter $\epsilon = 0.15$ and the target coverage rates $1-\alpha = 0.75, 0.85, 0.95$, and taking SCAN-BEST (ADADT-P) as an example, Fig.~\ref{fig:size_of_calibration_dataset} presents key statistical characteristics of the achieved coverage rate and candidate beam subset size for different calibration dataset sizes.
The box plots, showing median (horizontal line), first inter-quartile intervals (box), and support (whiskers) are evaluated with 15 independent experiments.
The results show that, regardless of the dataset size, the average achieved coverage rate remains above the target coverage rate. 
Larger dataset sizes offer more reliable guarantees for the coverage rate and a more stable size of the candidate beam subset.
Finally, higher target coverage rates generally demand smaller calibration datasets.
For instance, for the same calibration dataset size $N_{\text{cal}}$, the inter-quartile ranges of achieved coverage rates at $1-\alpha=0.95$ are narrower than those at $1-\alpha=0.75$ and $1-\alpha=0.85$.

\subsection{Impact of Sub-6G System Parameters}
\subsubsection{Impact of Sub-6G Channel Estimation Power}
\begin{figure*}[htbp]
    \centering 
    \includegraphics[width=0.98\textwidth]{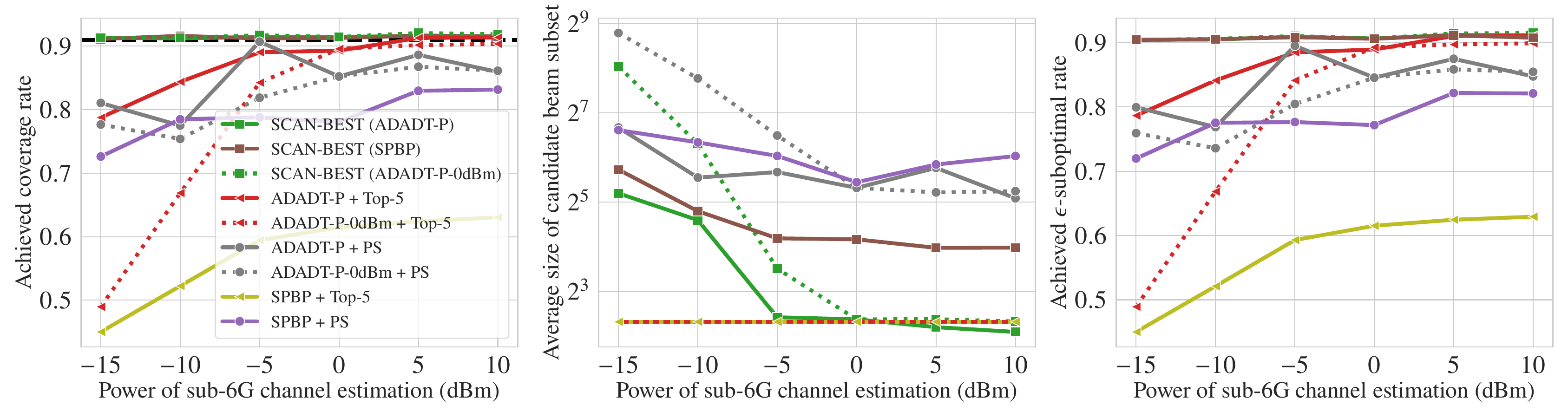}
    \caption{Achieved coverage rates, average sizes of candidate beam subset, and average suboptimality ratios for SCAN-BEST and baselines under different powers of sub-6G channel estimation.}
    \label{fig:sub_6G_power}
\end{figure*}
\begin{figure*}[htbp]
    \centering 
    \includegraphics[width=0.98\textwidth]{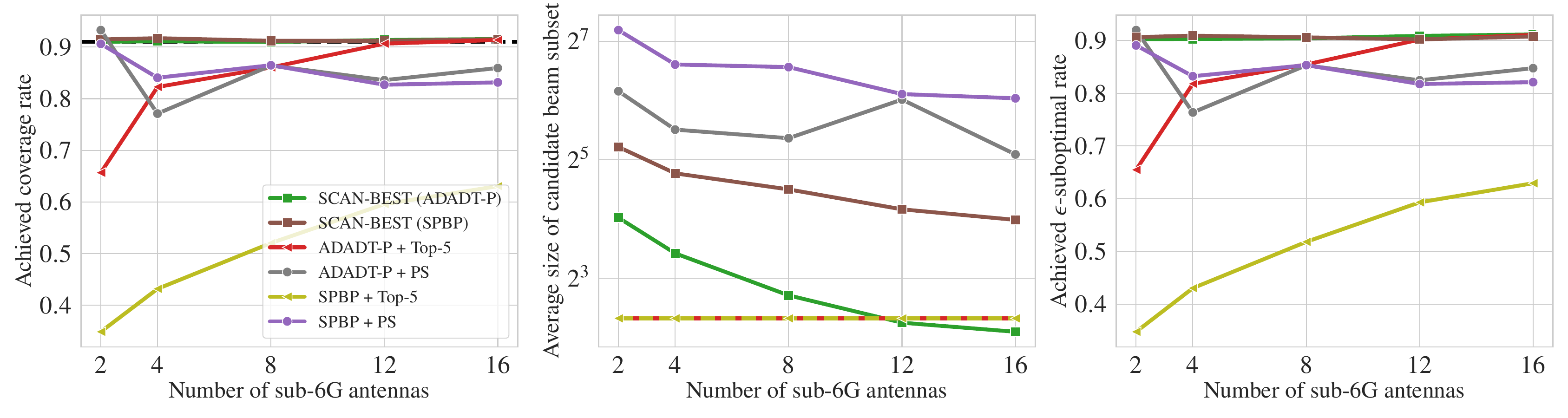}
    \caption{Achieved coverage rates, average sizes of candidate beam subset, and average suboptimality ratios for SCAN-BEST and baselines under different numbers of sub-6G antennas.}
    \label{fig:sub_6G_antennas}
\end{figure*}
Following the previous setting of $1-\alpha=0.91$ and $\epsilon=0.15$, we investigate the impact of the sub-6G channel estimation power, $\underline{P_{\text{s}}}$, on the performance of SCAN-BEST. 
Fig.~\ref{fig:sub_6G_power} shows that both SCAN-BEST (ADADT-P) and SCAN-BEST (SPBP) reliably maintain coverage rate guarantee $1-\alpha=0.91$, regardless of the power $\underline{P_{\text{s}}}$.
When $\underline{P_{\text{s}}}$ decreases, the candidate beam subset dynamically expands to satisfy the target coverage rate.
Furthermore, we assess the generalization capability\footnote{We have also evaluated its generalization capability to untrained scenarios (Appendix~\ref{app:unseen_scenario}) and its performance in outdoor scenarios (Appendix~\ref{app:outdoor_scenario}).}
of SCAN-BEST by applying ADADT-P trained at $\underline{P_{\text{s}}}=0~\text{dBm}$, denoted by ``SCAN-BEST (ADADT-P-0dBm)'', to test on datasets characterized by different power $\underline{P_{\text{s}}}$.
It exhibits similar performance to SCAN-BEST, which is attributed to the fact that CRC is a model-free calibration approach.
In contrast, baselines using the Top-$K$ and PS candidate beam subsets inevitably experience performance degradation or fluctuation as the power $\underline{P_{\text{s}}}$ varies.

\subsubsection{Impact of the Number of Sub-6G Antennas}
Similarly, under the setting of $1-\alpha=0.91$ and $\epsilon = 0.15$, SCAN-BEST (ADADT-P) and SCAN-BEST (SPBP) provide a reliable coverage rate guarantee of 0.91, regardless of the number of antennas $\underline{N_{\text{t}}}$, as shown in Fig.~\ref{fig:sub_6G_antennas}.
In scenarios with low number of antennas $\underline{N_{\text{t}}}$, which leads to poor angle resolution, both SCAN-BEST (ADADT-P) and SCAN-BEST (SPBP) dynamically expand their candidate beam subsets. 
However, SCAN-BEST (ADADT-P) maintains smaller candidate beam subsets due to the efficiency of the pre-processed ADADT.

\subsection{Impact of Statistical Shifts between Calibration and Test Datasets}
\begin{figure}[htbp]
\centering 
\includegraphics[width=0.48\textwidth]{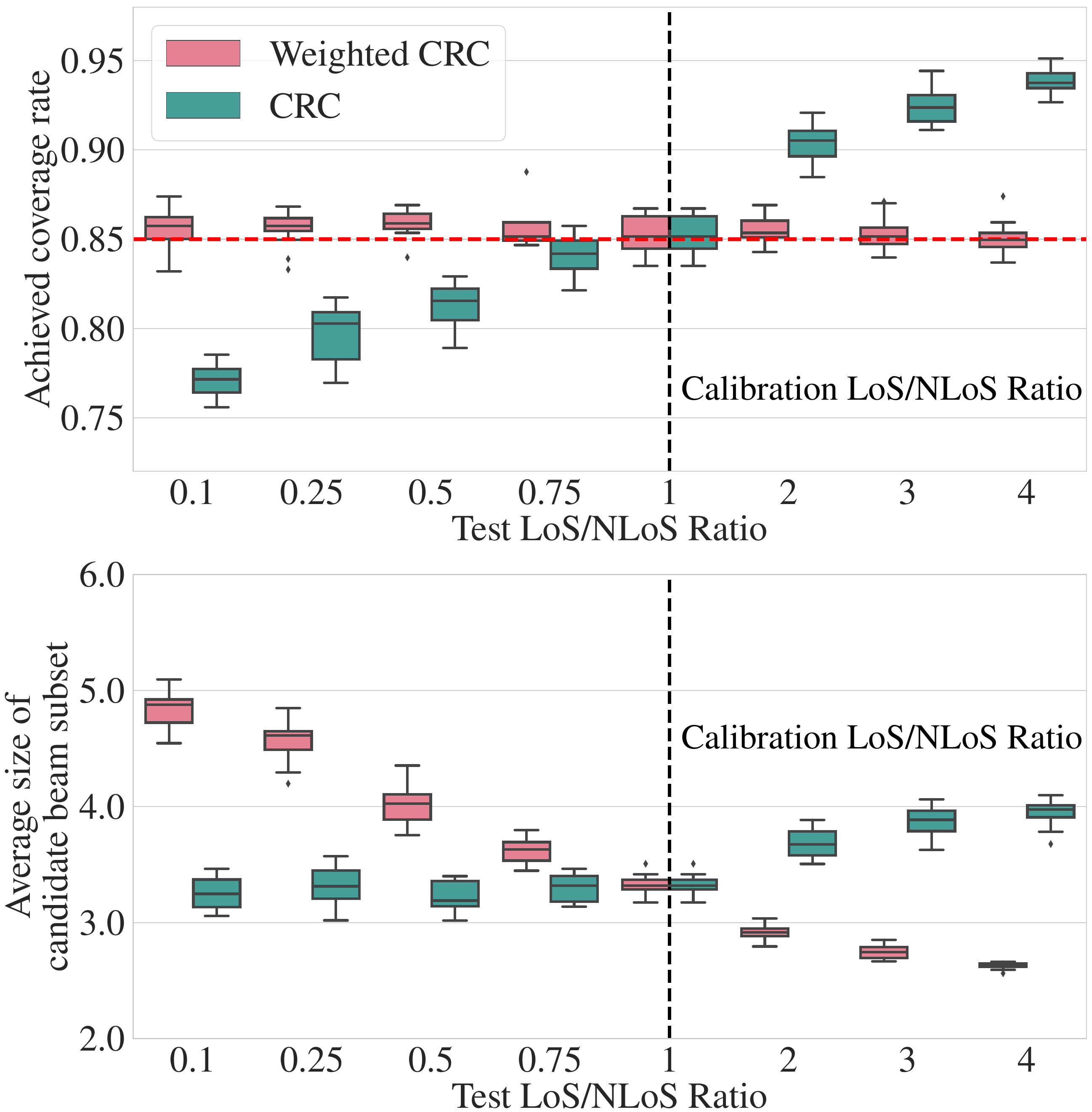}
\caption{Box plots of the achieved coverage rate and the size of the candidate beam subset under different LoS/NLoS ratios in the test data, with the LoS/NLoS ratio in the calibration data fixed at 1 (black line). The red lines denote the target coverage rates $1-\alpha$.}
\label{fig:weighted_CRC}
\end{figure}
In this subsection, we simulate a scenario in which the calibration and test data differ in the probability of LoS connectivity to validate the advantages of weighted CRC as compared to conventional CRC.
Fixing the suboptimality parameter $\epsilon = 0.15$ and the target coverage rate $1-\alpha = 0.85$, and taking ADADT-P as the predictor, Fig.~\ref{fig:weighted_CRC} presents the achieved coverage rate and the size of the candidate beam subset for CRC and weighted CRC as a function of the LoS/NLoS ratio in the test distribution, when the LoS/NloS ratio in the calibration distribution is fixed at 1.
For a LoS probability $p^{\text{LoS}}$, the LoS/NLoS ratio is defined as $p^{\text{LoS}}/(1-p^{\text{LoS}})$.
All box plots are evaluated over 15 independent experiments.

\par \textcolor{black}{As in \cite{tibshirani2019conformal}, the probabilistic classifier $g(\cdot)$ is trained by minimizing the cross-entropy loss between the predicted probabilities and the corresponding binary labels, where a label of 1 indicates that the sample originates from the test distribution, and a label of 0 indicates that it is from the calibration distribution.
The detailed configurations and training settings are provided in Appendix~\ref{appendix:C}.}
The results demonstrate that, regardless of the degree of misalignment between the calibration and test data, the weighted CRC method can effectively adapt to the test distribution, consistently achieving an average coverage rate higher than the target.
In particular, when the proportion of LoS samples is large, weighted CRC provides reliable coverage guarantees and achieves a more stable candidate beam subset size than for lower probability, as the mapping from sub-6G information to mmWave beams becomes less uncertain.
In contrast, CRC struggles to provide effective coverage in the covariate shift cases.
Specifically, when the LoS/NLoS ratio is low, CRC fails to meet the coverage requirement. 
Conversely, when the LoS/NLoS ratio is high, it tends to produce an overly large candidate beam subset, which is inefficient for achieving the target coverage rate.

\section{Conclusion}
\label{sec:conclusion}
In this paper, we have proposed SCAN-BEST, a novel theoretically principled framework to enhance near-field beam selection while ensuring a guaranteed coverage rate.
It wraps around any deep learning best-beam predictor, capturing the spatial-temporal correlation between sub-6G and mmWave channels and predicting of beam probabilities from sub-6G data. 
A CRC-based module, optionally weighted to cope with covariate shift, constructs a candidate beam subset with formal guarantees, which is then refined via limited mmWave beam training.
Extensive simulations show that SCAN-BEST ensures statistical suboptimality requirements and is robust to various sub-6G configurations. 
Future work may explore CRC-based beamforming with multi-modal sensing in high-mobility settings.

\appendices
\section{Proof of Proposition \ref{prop: empirical risk function} and Computation of Threshold $\hat{\lambda}$}
\subsection{Proof of Proposition \ref{prop: empirical risk function}}
\label{appendix:A}
Given the pretrained distribution $\mathscr{P}_{\boldsymbol{\theta}^{\star}}(\widehat{\underline{\mathbf{H}}})$ and the definition of the candidate beam subset $\mathcal{C}_{\lambda}(\widehat{\underline{\mathbf{H}}})$ in \eqref{def candidate beam subset}.
The miscoverage probability requirement in \eqref{target coverage rate} can be expressed as 
\begin{equation}
    \mathbf{Pr}\left(\nexists \mathbf{w} \in \mathcal{C}_{\lambda}(\widehat{\underline{\mathbf{H}}}) : r(\mathbf{w}, \mathbf{H}) \geqslant 1 - \epsilon \right)  \leqslant \alpha, \label{eq:miscoverage probability requirement}
\end{equation}
where the probability is taken over the underlying distribution of $\{\widehat{\underline{\mathbf{H}}},\mathbf{H}\}$.

\par We define a loss function $\ell(\widehat{\underline{\mathbf{H}}},\mathbf{H},\lambda)\in \{0,1\}$ as 
\begin{equation}
\ell(\widehat{\underline{\mathbf{H}}},\mathbf{H},\lambda)=\mathds{1}\left(\nexists \mathbf{w} \in \mathcal{C}_{\lambda}(\widehat{\underline{\mathbf{H}}}): r(\mathbf{w},\mathbf{H})\geqslant 1-\epsilon\right). \label{loss function}
\end{equation}
As a result, the miscoverage probability requirement in \eqref{target coverage rate} is equivalent to
\begin{equation}
\mathbf{Pr}\left(\ell(\widehat{\underline{\mathbf{H}}},\mathbf{H},\lambda)=1\right)\leqslant \alpha,
\end{equation}
which can be expressed as
\begin{equation}
    \label{risk control problem}
    \mathbb{E}\left[\ell(\widehat{\underline{\mathbf{H}}},\mathbf{H},\lambda)\right] \leqslant \alpha.
\end{equation} 
Thus, the miscoverage probability requirement in \eqref{target coverage rate} is equivalent to the risk control problem in \eqref{risk control problem}.
\par The empirical risk function $\hat{R}(\lambda)$ in \eqref{eq: empirical risk function} can be written as 
\begin{equation}
\hat{R}(\lambda)=\frac{1}{N_{\text{cal}}} \sum_{i=1}^{N_{\text{cal}}} \ell(\widehat{\underline{\mathbf{H}}}_i,\mathbf{H}_i,\lambda). \label{eq: empirical risk function loss form}
\end{equation}
With a little abuse of notation, we denote $\ell(\widehat{\underline{\mathbf{H}}}_i,\mathbf{H}_i,\lambda)$ as $\ell_{i}(\lambda)$, hence $\hat{R}(\lambda)$ can be expressed as $\hat{R}(\lambda)=\sum_{i=1}^{N_{\text{cal}}} \ell_{i}(\lambda)/N_{\text{cal}}$.
Here, we introduce an auxiliary risk function $\hat{R}_{\text{f}}(\lambda)$ accounting for both calibration samples $\{\widehat{\underline{\mathbf{H}}}_i,\mathbf{H}_i\}_{i=1}^{N_{\text{cal}}}$ and a future sample $\{\widehat{\underline{\mathbf{H}}}_{N_{\text{cal}}+1},\mathbf{H}_{N_{\text{cal}}+1}\}$, which is given by 
\begin{equation}
    \hat{R}_{\text{f}}(\lambda)=\frac{1}{N_{\text{cal}}+1} \sum_{i=1}^{N_{\text{cal}}+1} \ell_i(\lambda).
\end{equation}
\begin{proposition}
The functions $\ell_i(\lambda)$ is monotonically non-increasing w.r.t. $\lambda$.
\label{prop: monotonically non-increasing}
\end{proposition}
\begin{proof}
For an arbitrary pair $\{\widehat{\underline{\mathbf{H}}}_i,\mathbf{H}_i\}$, assume $\lambda_1 < \lambda_2$, and consider two candidate beam subsets $\mathcal{C}_{\lambda_1}(\widehat{\underline{\mathbf{H}}}_i)=\{\mathbf{w}_{n,s}\in\mathcal{W}:u(n,s,\widehat{\underline{\mathbf{H}}}_i)\leqslant \lambda_1\}$ and $\mathcal{C}_{\lambda_2}(\widehat{\underline{\mathbf{H}}}_i)=\{\mathbf{w}_{n,s}\in\mathcal{W}:u(n,s,\widehat{\underline{\mathbf{H}}}_i)\leqslant \lambda_2\}$.
For any $\mathbf{w}_{n,s} \in \mathcal{C}_{\lambda_1}(\widehat{\underline{\mathbf{H}}}_i)$, it follows $u(n,s,\widehat{\underline{\mathbf{H}}}_i)\leqslant \lambda_1\leqslant \lambda_2$, which implies $\mathbf{w}_{n,s} \in \mathcal{C}_{\lambda_2}(\widehat{\underline{\mathbf{H}}}_i)$.
Hence, we have the implication
\begin{equation}
\lambda_1<\lambda_2 \Longrightarrow \mathcal{C}_{\lambda_1}(\widehat{\underline{\mathbf{H}}}_i)\subset \mathcal{C}_{\lambda_2}(\widehat{\underline{\mathbf{H}}}_i). \notag
\end{equation}
It can be also seen that the inequality $\ell_i(\lambda_1)\geqslant \ell_i(\lambda_2)$ holds.
\end{proof}

\par We also have
\begin{align}
\hat{R}_{\text{f}}(\lambda)&=\frac{N_{\text{cal}}}{N_{\text{cal}}+1}\hat{R}(\lambda)+\frac{\ell_{N_{\text{cal}}+1}(\lambda)}{N_{\text{cal}}+1}\notag\\
&\mathop{\leqslant}^{(a)} \frac{N_{\text{cal}}}{N_{\text{cal}}+1}\hat{R}(\lambda)+\frac{1}{N_{\text{cal}}+1}, \label{eq: risk function inequality}
\end{align}
where condition (a) holds because $\ell_{N_{\text{cal}}+1}(\lambda)\in \{0,1\}\leqslant 1$.
Then, using the optimized $\hat{\lambda}$ in \eqref{eq: hat lambda optimization}, we obtain that 
\begin{equation}
\frac{N_{\text{cal}}}{N_{\text{cal}}+1}\hat{R}(\hat{\lambda})+\frac{1}{N_{\text{cal}}+1}\leqslant \alpha.
\end{equation}
According to the inequality in \eqref{eq: risk function inequality}, this implies that 
\begin{gather}
    \hat{R}_{\text{f}}(\hat{\lambda})\leqslant \frac{N_{\text{cal}}}{N_{\text{cal}}+1}\hat{R}(\hat{\lambda})+\frac{1}{N_{\text{cal}}+1}\leqslant \alpha \implies \notag\\
\hat{R}_{\text{f}}(\hat{\lambda})\leqslant \alpha.
\end{gather}
Introducing $\hat{\lambda}^\downarrow =\inf\{\lambda\in \mathbb{R}: \hat{R}_{\text{f}}(\lambda)\leqslant \alpha\}$, we know that $\hat{\lambda}^\downarrow\leqslant \hat{\lambda} $.
Thus, according to Proposition~\ref{prop: monotonically non-increasing}, we obtain
\begin{equation}
    \mathbb{E}[\ell_{N_{\text{cal}}+1}(\hat{\lambda})]\leqslant \mathbb{E}[\ell_{N_{\text{cal}}+1}(\hat{\lambda}^\downarrow)]. \label{eq: inequality_hat_lambda}
\end{equation}

In mathematics, a bag (also called a multiset) is a collection where elements can appear more than once.
\begin{lemma}[From \cite{tibshirani2019conformal}]
If the random variables $x_1, \ldots, x_n,\allowbreak x_{n+1}$ are exchangeable, and conditioned on the bag of realization of $\{x_1, \ldots, x_n, x_{n+1}\}$, the random variable $x_{n+1}$ is uniformly distributed in the set $\{x_1, \ldots, x_n, x_{n+1}\}$.
\end{lemma}
Since the calibration samples $\{\widehat{\underline{\mathbf{H}}}_i,\mathbf{H}_i\}_{i=1}^{N_{\text{cal}}}$ and the future sample $\{\widehat{\underline{\mathbf{H}}}_{N_{\text{cal}}+1},\mathbf{H}_{N_{\text{cal}}+1}\}$ are i.i.d. from the same scenario, they are exchangeable.
According to the exchangeability-preserving theorem (Theorem 3 in \cite{kuchibhotla2020exchangeability}), the functions $\ell_1(\lambda),\dots,\ell_{N_{\text{cal}}}(\lambda),\ell_{N_{\text{cal}}+1}(\lambda)$ are also exchangeable, and we obtain 
\begin{equation}
    \ell_{N_{\text{cal}}+1}(\lambda)\sim\text{Uniform}\left(\ell_1(\lambda),\dots,\ell_{N_{\text{cal}}}(\lambda),\ell_{N_{\text{cal}}+1}(\lambda)\right).
\end{equation}
Conditioned on the bag of realization of $\{\ell_1(\lambda),\dots,\ell_{N_{\text{cal}}}(\lambda),\allowbreak \ell_{N_{\text{cal}}+1}(\lambda)\}$, when $\lambda=\hat{\lambda}^\downarrow$, we have 
\begin{gather}
    \mathbb{E}\left[\ell_{N_{\text{cal}}+1}(\hat{\lambda}^\downarrow)\right]=\frac{1}{N_{\text{cal}}+1}\sum\limits_{i=1}^{N_{\text{cal}}+1}\ell_i(\hat{\lambda}^\downarrow)=\hat{R}_{\text{f}}(\hat{\lambda}^\downarrow)\mathop{\leqslant}\limits^{(b)} \alpha \implies \notag \\
     \mathbb{E}[\ell_{N_{\text{cal}}+1}(\hat{\lambda}^\downarrow)]\leqslant \alpha,
\end{gather}
where the condition (b) holds according to the definition of $\hat{\lambda}^\downarrow$.
Finally, according to \eqref{eq: inequality_hat_lambda}, we obtain the following result:
\begin{equation}
    \mathbb{E}[\ell_{N_{\text{cal}}+1}(\hat{\lambda})]\leqslant \alpha.
\end{equation}

\subsection{A Low Complexity Algorithm to Determine $\hat{\lambda}$}
\label{appendix: hat lambda}
Recall that $\hat{R}(\lambda) = \frac{1}{N_{\text{cal}}}\sum_{i=1}^{N_{\text{cal}}} \ell(\widehat{\underline{\mathbf{H}}}_i, \mathbf{H}_i, \lambda)$ as in~(34).
For any given pair $\{\widehat{\underline{\mathbf{H}}}_i, \mathbf{H}_i\}$, the function $\ell(\widehat{\underline{\mathbf{H}}}_i, \mathbf{H}_i, \lambda)$ is an indicator function w.r.t. $\lambda$.
For each such pair in the calibration dataset, we define $\lambda_i$ as the minimum $\lambda$ such that $\ell(\widehat{\underline{\mathbf{H}}}_i,\mathbf{H}_i,\lambda)=0$, i.e.,
\begin{equation}
\ell(\widehat{\underline{\mathbf{H}}}_i,\mathbf{H}_i,\lambda)=\mathds{1}(\lambda<\lambda_i)=
\begin{cases}
    1,& \lambda <\lambda_i, \\
    0,& \text{if } \lambda \geqslant \lambda_i
\end{cases}.
\end{equation}
Specifically, the value of $\lambda_i$ can be obtained as
\begin{equation}
\lambda_i=\inf\left\{u(n,s,\underline{\widehat{\mathbf{H}}}_i):\mathbf{w}_{n,s}\in \mathcal{F}_{\epsilon}(\mathbf{H}_i)\right\},
\end{equation}
where $\mathcal{F}_{\epsilon}(\mathbf{H}_i)=\{\mathbf{w}\in\mathcal{W}:r({\mathbf{w}},\mathbf{H}_i) \geqslant 1-\epsilon\}$ is the set of beam candidates that achieve sufficiently high performance with respect to $\mathbf{H}_i$.

\par We thus obtain that 
\begin{equation}
\hat{R}(\lambda)=\frac{1}{N_{\text{cal}}}\sum\limits_{j=1}^{N_{\text{cal}}}\mathds{1}(\lambda< \lambda_i)=1-\frac{1}{N_{\text{cal}}}\sum\limits_{j=1}^{N_{\text{cal}}}\mathds{1}(\lambda\geqslant \lambda_i).
\end{equation}
That is, $1-\hat{R}(\lambda)$ represents the empirical cumulative distribution function of the set $\{\lambda_j\}_{j=1}^{N_{\text{cal}}}$, and can be interpreted as the normalized rank of $\lambda$ within this set. 
To determine $\hat{\lambda}$ that satisfies the condition in~(22), we first sort $\{\lambda_j\}_{j=1}^{N_{\text{cal}}}$ in an ascending order, and $\hat{\lambda}$ is then selected as the $(1-\frac{\left\lfloor \alpha N_{\text{cal}} + \alpha - 1 \right\rfloor}{N_{\text{cal}}})$-quantile of sorted values, i.e.,
\begin{equation}
    \hat{\lambda}= \text{quantile}\left(1-\frac{\left\lfloor \alpha N_{\text{cal}}+\alpha - 1 \right\rfloor}{N_{\text{cal}}};\{\lambda_i\}_{i=1}^{N_{\text{cal}}}\right).
\end{equation}
This quantile-based procedure enables efficient selection of $\hat{\lambda}$ to meet the desired risk level $\alpha$.

\section{Proof of Proposition \ref{prop: weighted_CRC}}
\label{appendix:B}
For the samples in the calibration dataset, we denote their loss as $\ell_i(\lambda)$ for $i=1,...,N_{\text{cal}}$, while, for the future test sample $\{\widehat{\underline{\mathbf{H}}}',\mathbf{H}'\}$, we denote its loss as $\ell'(\lambda)$. 
We have $\tilde{R}(\lambda)=\sum\nolimits_{i=1}^{N_{\text{cal}}}\omega_i\ell_i(\lambda)$.
Define an auxiliary risk function $\tilde{R}_{\text{f}}(\lambda)$ accounting for both calibration samples $\{\widehat{\underline{\mathbf{H}}}_i,\mathbf{H}_i\}_{i=1}^{N_{\text{cal}}}$ and the future test sample $\{\widehat{\underline{\mathbf{H}}}',\mathbf{H}'\}$ as
\begin{eqnarray}
    \tilde{R}_{\text{f}}(\lambda)=\sum_{i=1}^{N_{\text{cal}}}\omega_i\ell_i(\lambda)+\omega'\ell'(\lambda).
\end{eqnarray}

\par Accordingly, we obtain the inequality
\begin{equation}
\tilde{R}_{\text{f}}(\lambda)\mathop{\leqslant}^{(a)} \tilde{R}(\lambda)+\omega', \label{eq: weighted risk function inequality}
\end{equation}
where condition (a) holds because $\ell'(\lambda)\in \{0,1\}\leqslant 1$.
When using the optimized $\hat{\lambda}$ in \eqref{eq: weighted CRC threshold}, we have 
\begin{equation}
    \tilde{R}(\hat{\lambda})+\omega'\leqslant \alpha.
\end{equation}
According to the inequality in \eqref{eq: weighted risk function inequality}, this indicates that
\begin{equation}
\tilde{R}_{\text{f}}(\lambda)\leqslant \alpha.
\end{equation}
Introducing $\hat{\lambda}^\downarrow=\inf\{\lambda\in \mathbb{R}: \tilde{R}_{\text{f}}\leqslant \alpha\}$, it implies $\hat{\lambda}^\downarrow\leqslant \hat{\lambda}$.
Then, according to Proposition~\eqref{prop: monotonically non-increasing}, we conclude
\begin{equation}
    \mathbb{E}[\ell'(\hat{\lambda})]\leqslant \mathbb{E}[\ell'(\hat{\lambda}^\downarrow)]. \label{eq: weighted_inequality_hat_lambda}
\end{equation}

\begin{lemma}[Lemma 3 of \cite{tibshirani2019conformal}]
Let $z = (x, y)$. Suppose the random variables $z_1, \dots, z_n$ are from the calibration distribution $\Gamma_{\textup{cal}}(z)$, and $z'$ is from the test distribution $\Gamma_{\textup{te}}(z)$. 
$z_1,\dots,z_n,z'$ are conditioned on the bag of realization of $\{z_1,\dots,z_n,z'\}$.
Define the weights
\begin{equation}
\omega_i = \frac{p_i}{\sum_{j=1}^{n} p_j + p'}, \quad \omega' = \frac{p'}{\sum_{j=1}^{n} p_j + p'},
\end{equation}
where
\begin{equation}
p_i = \frac{\Gamma_{\textup{te}}(z_i)}{\Gamma_{\textup{cal}}(z_i)}, \quad p' = \frac{\Gamma_{\textup{te}}(z')}{\Gamma_{\textup{cal}}(z')},
\end{equation}
for $i = 1, \dots, n$. Then, the test sample $z'$ is distributed as
\begin{equation}
z' \sim \sum_{i=1}^{n} \omega_i \delta_{z_i} + \omega' \delta_{z'},
\end{equation}
where $\delta_z$ denotes the Dirac measure at $z$. Moreover, for any non-random loss function $\ell(\cdot)$ over $z$, it follows that
\begin{equation}
\ell(z') \sim \sum_{i=1}^{n} \omega_i \delta_{\ell(z_i)} + \omega' \delta_{\ell(z')}.
\end{equation}
\end{lemma}

\par In the considered scenario, the conditional distribution of $\mathbf{H}$ given $\widehat{\underline{\mathbf{H}}}$ is assumed to remain the same in the calibration and test datasets \cite{alrabeiah2020deep}. 
Under this covariate shift assumption, the weights $\omega_1, \dots, \omega'$ can be computed as in \eqref{eq:importance weight}. 
Based on the loss function $\ell_{\lambda}(\cdot)$ defined in \eqref{loss function}, and setting $\lambda = \hat{\lambda}^\downarrow$, we have
\begin{gather}
    \ell'(\hat{\lambda}^\downarrow) \sim \sum_{i=1}^{n} \omega_i \delta_{\ell_i(\hat{\lambda}^\downarrow)} + \omega' \delta_{\ell'(\hat{\lambda}^\downarrow)} \implies \notag \\
    \mathbb{E}\left[\ell'(\hat{\lambda}^\downarrow)\right] = \sum_{i=1}^{n} \omega_i \ell_i(\hat{\lambda}^\downarrow) + \omega' \ell'(\hat{\lambda}^\downarrow) = \tilde{R}_{\text{f}}(\hat{\lambda}^\downarrow) \mathop{\leqslant}\limits^{(b)} \alpha \implies \notag\\
    \mathbb{E}\left[\ell'(\hat{\lambda}^\downarrow)\right] \leqslant \alpha,
\end{gather}
where condition (b) follows from the definition of $\hat{\lambda}^\downarrow$.
Finally, by the inequality in \eqref{eq: weighted_inequality_hat_lambda}, we conclude:
\begin{equation}
    \mathbb{E}[\ell'(\hat{\lambda})] \leqslant \alpha.
\end{equation}

\section{Additional Experimental Materials}
\subsection{Network Configurations and Training Parameters}
\label{appendix:C}
\begin{figure*}[t]
    \centering
    \setlength{\abovecaptionskip}{1pt} 
    \subfloat[$\epsilon=0.20$]{\includegraphics[width=0.30\textwidth]{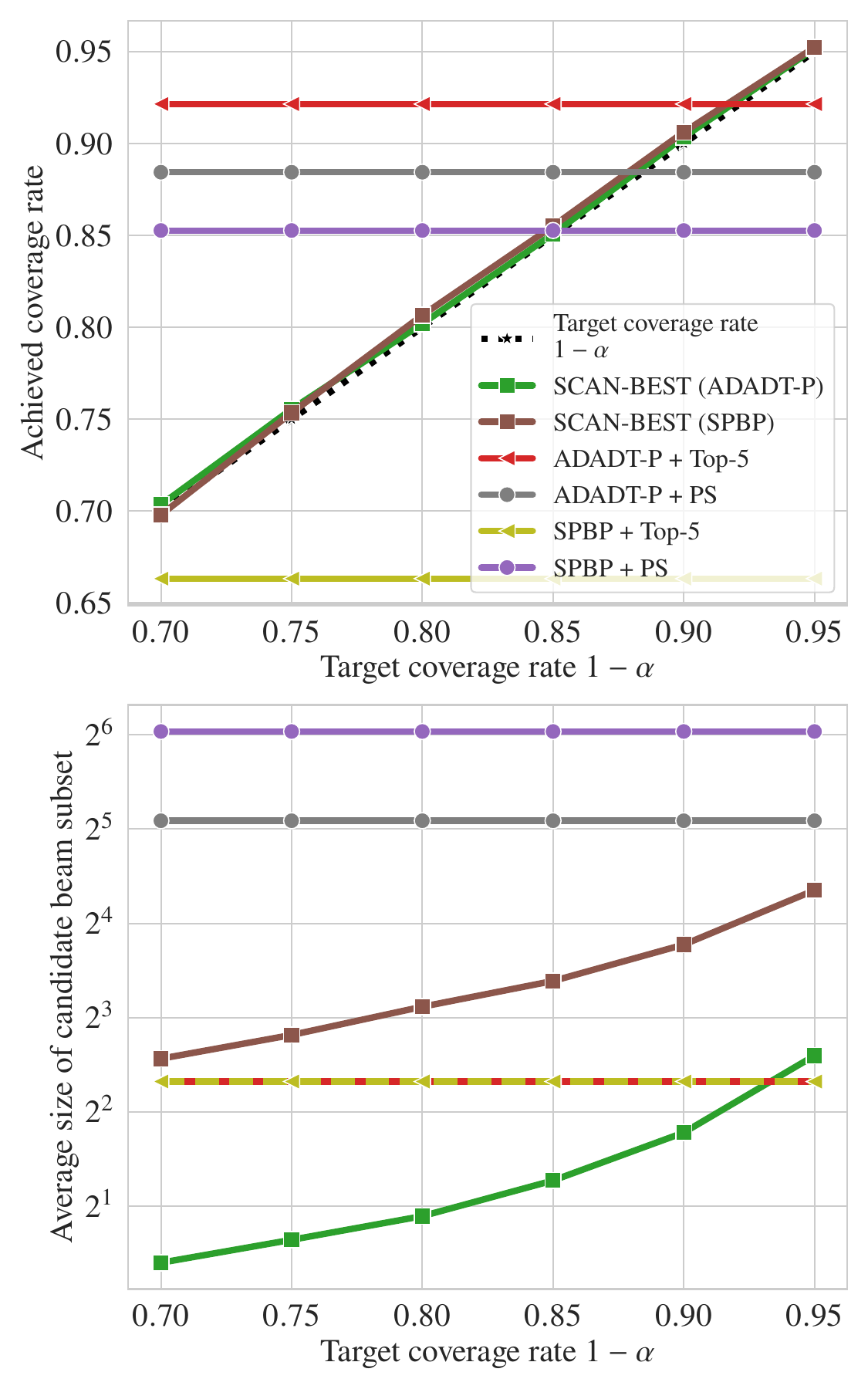}}
    \subfloat[$\epsilon=0.10$]{\includegraphics[width=0.30\textwidth]{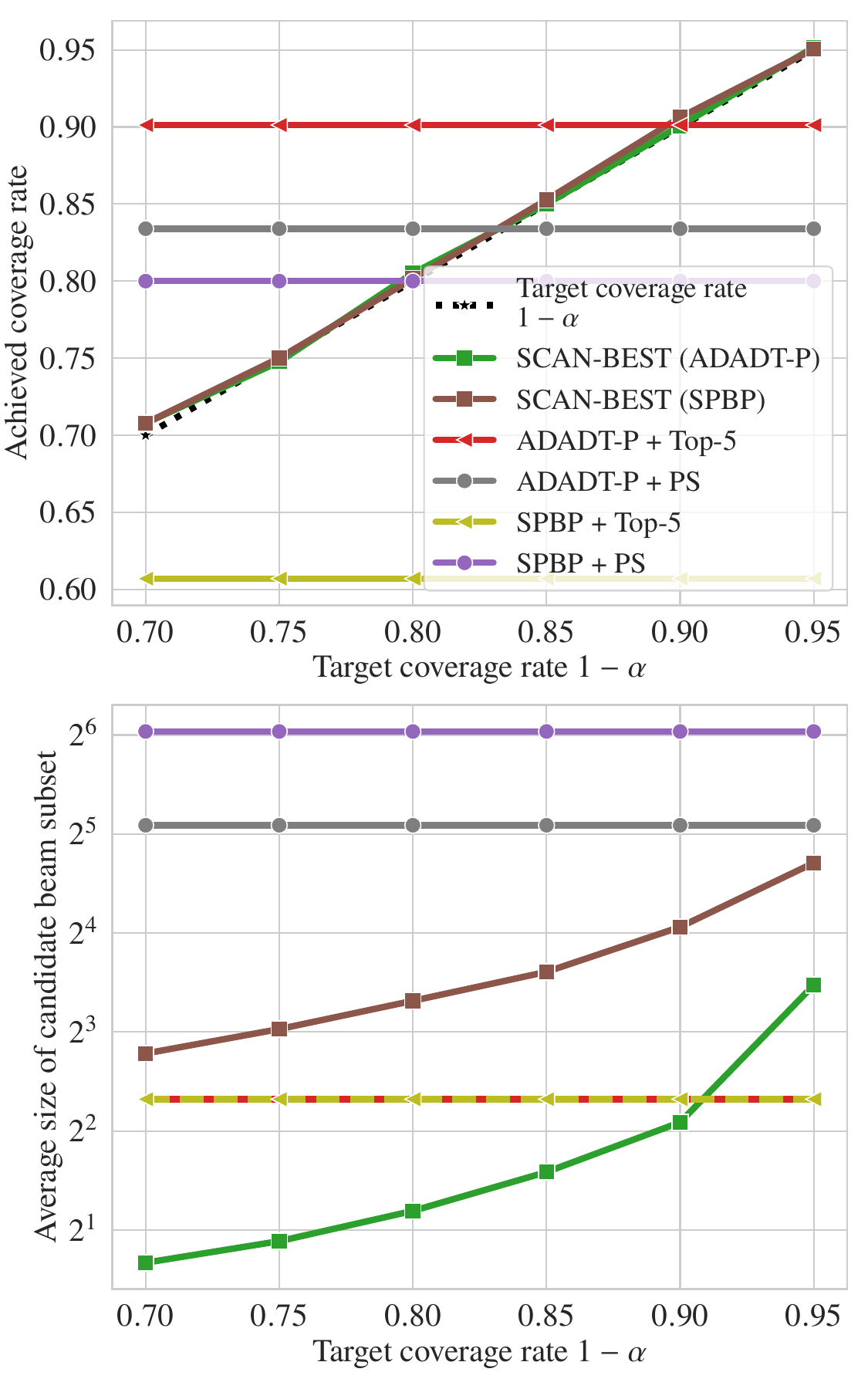}}\vspace{-6pt}\\
    \subfloat[$\epsilon=0.05$]{\includegraphics[width=0.30\textwidth]{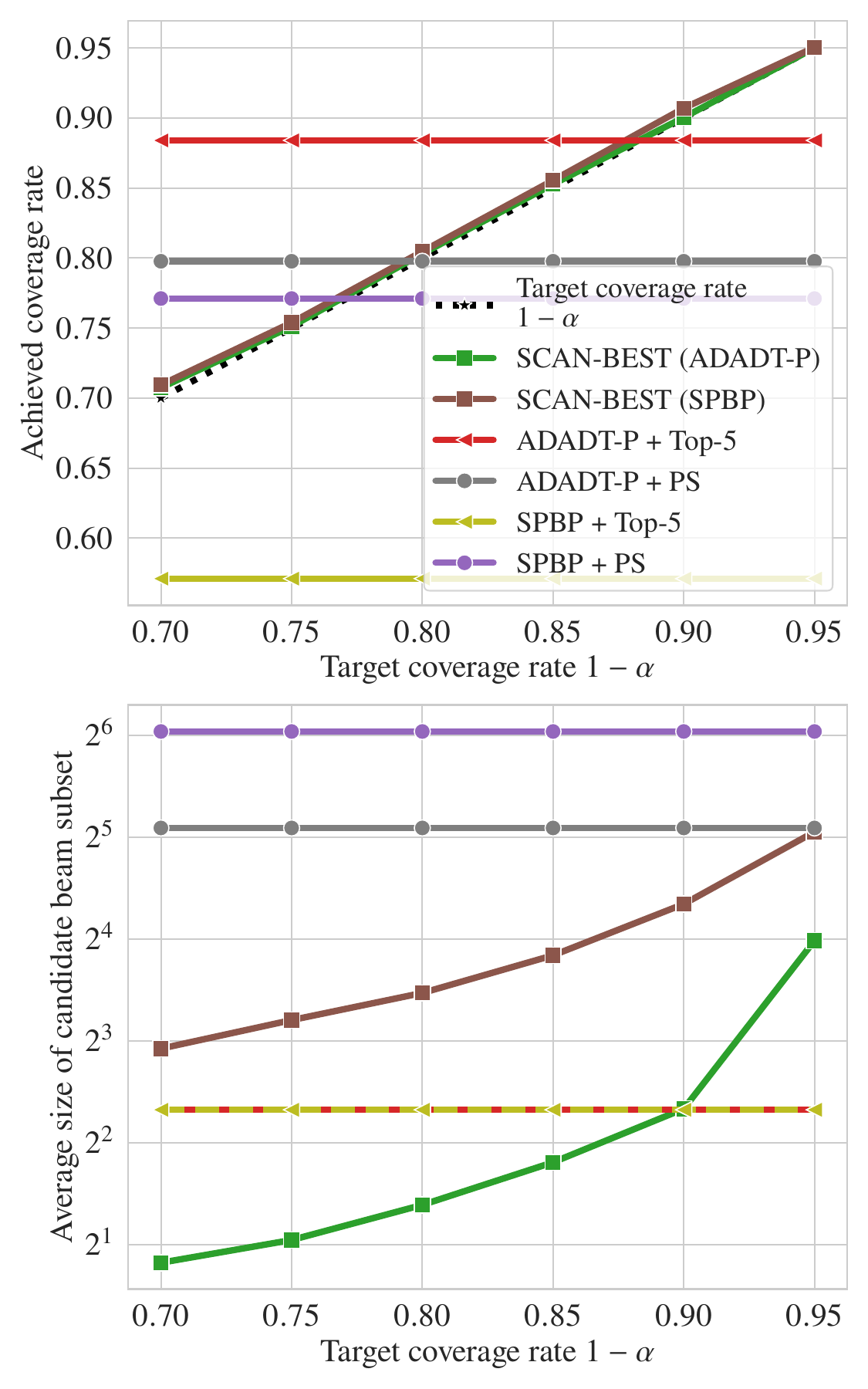}}
    \subfloat[$\epsilon=0.00$]{\includegraphics[width=0.30\textwidth]{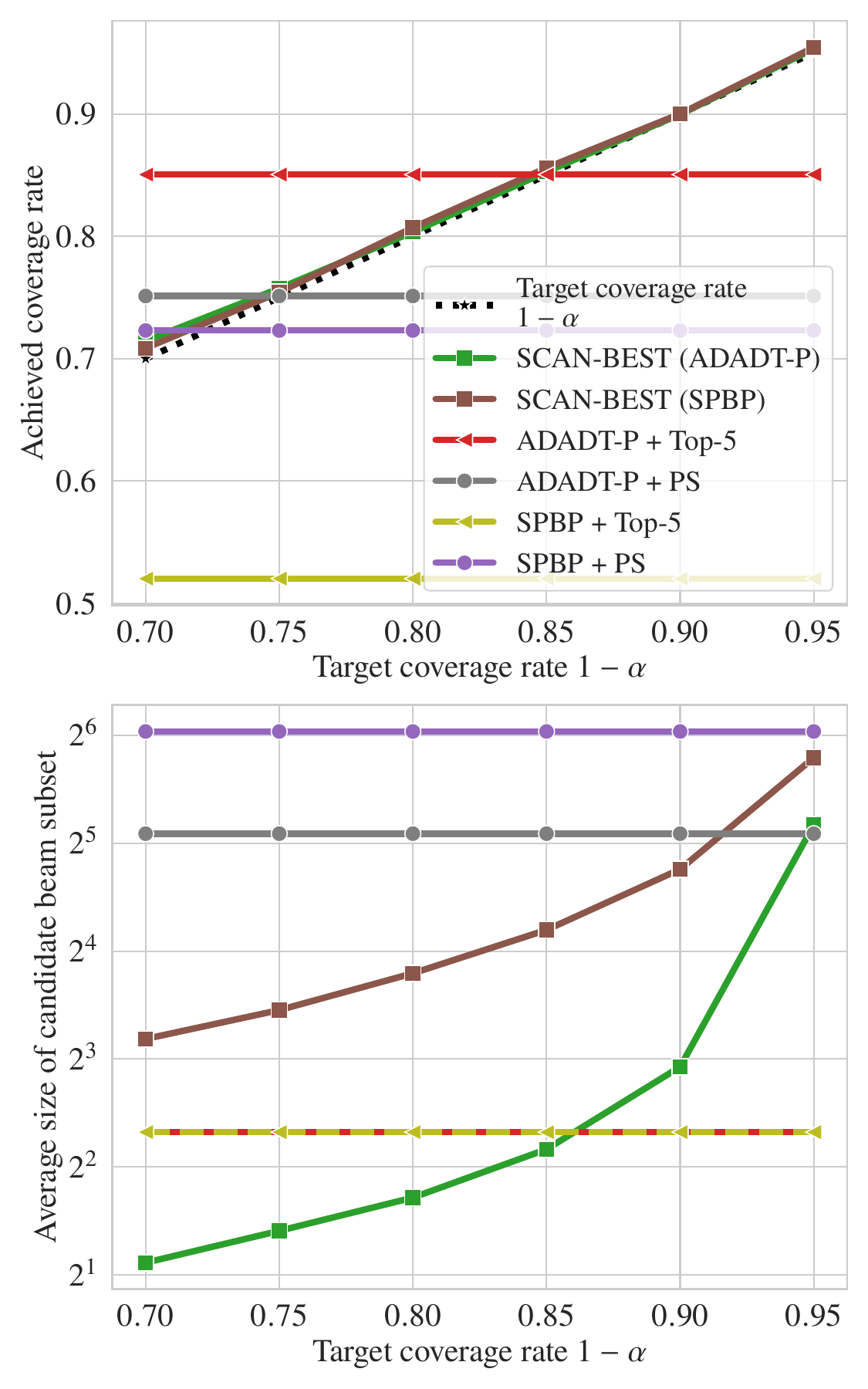}}
    \caption{Achieved coverage rates and average sizes of candidate beam subset of SCAN-BEST and baselines for $\epsilon=0.20,0.10,0.05,0.00$.}
    \label{fig:coverage_and_size:part2}
\end{figure*}
To balance performance and complexity, both the predictors ADADT-P and SPBP are implemented as 2D CNN architectures, whose parameters are provided in Tables~\ref{ADADT-P parameters} and \ref{SPBP parameters}, respectively.
Both of them are trained using the Adam optimizer with a batch size of 128, an initial learning rate of 0.0002, and the learning rate scheduler ``ReduceLROnPlateau''~\cite{noauthor_reducelronplateau_nodate}.
This scheduler automatically reduces the learning rate when the model's performance on the validation set ceases to improve or demonstrates only marginal improvements over a specified number of training epochs.
Both the maximum training epochs of ADADT-P and SPBP are set to 200, and, to prevent overfitting, a classical early stopping criterion is adopted, which stops the training process if the validation loss does not improve for a certain number of epochs. 
The probabilistic classifier $g(\cdot)$ used in the weighted CRC, is also implemented as a 2D CNN architecture with its configuration in Table~\ref{probabilistic classifier parameters}.
The optimizer, batch size, and learning rate schedule, are the same as for the predictors, and the initial learning rate is set as 0.0004.
\begin{table}[t]
    \renewcommand{\arraystretch}{0.9}
    \caption{ADADT-P Parameters}
    \label{ADADT-P parameters}
    \centering
    \setlength\tabcolsep{0.5mm}{
    \begin{tabular}{c|c|c}
    \toprule
    \makecell[c]{Input/Output \\Channel} & Layers & \makecell[c]{Convolution \\Kernels} \\
    \midrule
    \midrule
    (1,~32) & \scriptsize\makecell[c]{$\{\text{Conv2D, BatchNorm2D, ReLU}\}\times 2$,\\ MaxPool2D} & (8,~4)\\\hline
    (32,~64) & idem & (5,~3)\\\hline
    (64,~128) & idem & (5,~3)\\\hline
    (128,~256) & idem & (5,~3)\\\hline
    (256,~256) & \scriptsize Conv2D, BatchNorm2D & (3,~3)\\\hline
    (256,~128) & \scriptsize\makecell[c]{Upsample, Conv2D, BatchNorm2D, ReLU\\Conv2D, BatchNorm2D} & (3,~3)\\\hline
    (128,~64) & idem & (7,~3)\\\hline
    (64,~32) & idem & (7,~3)\\\hline
    (32,~16) & idem & (7,~3)\\\hline
    (16,~8) & idem & (7,~3)\\\hline
    (8,~1) & Conv2D & (7,~3)\\\hline
    / & Flatten, Softmax, Reshape & /\\
    \bottomrule
    \end{tabular}}
    \end{table}

    \begin{table}[t]
    \renewcommand{\arraystretch}{0.9}
    \caption{SPBP Parameters}
    \label{SPBP parameters}
    \centering
    \setlength\tabcolsep{0.5mm}{
    \begin{tabular}{c|c|c}
    \toprule
    \makecell[c]{Input/Output \\Channel} & Layers & \makecell[c]{Convolution \\Kernels} \\
    \midrule
    \midrule
    (1,~32) & \scriptsize\makecell[c]{$\{\text{Conv2D, BatchNorm2D, ReLU}\}\times 2$,\\ MaxPool2D} & (2,~4)\\\hline
    (32,~64) & idem & (2,~4)\\\hline
    (64,~128) & idem & (2,~4)\\\hline
    (128,~128) & idem & (2,~4)\\\hline
    (128,~64) & \scriptsize\makecell[c]{Upsample, Conv2D, BatchNorm2D, ReLU\\Conv2D, BatchNorm2D} & (3,~3)\\\hline
    (64,~32) & idem & (7,~3)\\\hline
    (32,~16) & idem & (7,~3)\\\hline
    (16,~8) & idem & (7,~3)\\\hline
    (8,~1) & Conv2D & (7,~3)\\\hline
    / & Flatten, Softmax, Reshape & /\\
    \bottomrule
    \end{tabular}}
        \end{table}

    \begin{table}[t]
    \linespread{1.0}\selectfont
    \caption{Probabilistic Classifier Parameters}
    \label{probabilistic classifier parameters}
    \centering
    \setlength\tabcolsep{0.5mm}{
    \begin{tabular}{c|c|c}
    \toprule
    \makecell[c]{Input/Output \\Channel} & Layers & \makecell[c]{Convolution \\Kernels} \\
    \midrule
    \midrule
    (1,~32) & \scriptsize\scriptsize\makecell[c]{$\text{Conv2D, BatchNorm2D, ReLU,}\text{MaxPool2D}$} & (3,~3)\\\hline
    (32,~64) & idem & (3,~3)\\\hline
    (64,~128) & idem & (3,~3)\\\hline
    (128,~256) & \scriptsize\scriptsize\makecell[c]{$\text{Conv2D, BatchNorm2D, ReLU,} $\\$\text{AdaptiveAvgPool2D}$} & (3,~3)\\\hline
    (256,~64) & Flatten, Linear, ReLU & /\\\hline
    (64,~1) & Linear, Sigmoid&/\\
    \bottomrule
    \end{tabular}}
    \end{table}

\par Moreover, in Fig.~\ref{fig:coverage_and_size:part2}, we present the additional results varying the suboptimality parameter $\epsilon$ value.

\subsection{Computational Complexity Analysis}
The computational complexity of the proposed SCAN-BEST framework primarily stems from two components: the deep learning-based predictor and the candidate beam subset selection based on CRC.
For the two predictors considered, ADADT-P and SPBP, their computational complexities are approximately $\mathcal{O}(N_{\text{t}}MB)$ and $\mathcal{O}(\underline{N_{\text{t}}}\underline{M}B)$, respectively, where $B$ denotes the number of various layers used in the predictor.
In the candidate beam subset selection stage, the threshold $\hat{\lambda}$ is computed offline during a calibration phase.
Therefore, the online selection only involves comparing $N_{\text{t}}S$ NSCs with $\hat{\lambda}$.
Thus, its computational complexity is approximately $\mathcal{O}(N_{\text{t}}S)$.
Therefore, SCAN-BEST (ADADT-P) and SCAN-BEST (SPBP) have computational complexities of $\mathcal{O}(N_{\text{t}}MB+N_{\text{t}}S)$ and $\mathcal{O}(\underline{N_{\text{t}}}\underline{M}B+N_{\text{t}}S)$, respectively.

\par As for the other baselines in Table~\ref{table:baseline_comparison}, the computational complexities of the Top-$K$ and PS-based selection schemes are $\mathcal{O}(K N_{\text{t}} S)$ and $\mathcal{O}(N_{\text{t}} S \log(N_{\text{t}} S))$, respectively. 
Thus, the total complexities of ADADT-P + Top-$K$, ADADT-P + PS, SPBP + Top-$K$, and SPBP + PS are $\mathcal{O}(N_{\text{t}}MB+KN_{\text{t}}S)$, $\mathcal{O}(N_{\text{t}}MB+N_{\text{t}}S\log(N_{\text{t}}S))$, $\mathcal{O}(\underline{N_{\text{t}}} \underline{M}B+KN_{\text{t}}S)$, and $\mathcal{O}(\underline{N_{\text{t}}} \underline{M}B+N_{\text{t}}S\log(N_{\text{t}}S))$, respectively.
For CSLW-NBS, its complexity is approximately $\mathcal{O}(IAN_{\text{t}}M + IAeM + I^2Ae^2)$, where $I$, $A$, and $e$ denote the number of BOMP iterations, the number of measurements, and the block size, respectively.

\par In summary, the computational complexity of SCAN-BEST is slightly lower than that of the baseline methods that combine predictors with alternative candidate beam subset selection schemes, and significantly lower than that of CSLW-NBS. 
In practical implementations, with GPU acceleration for the CNN components, the runtime of SACN-BEST is very short—and substantially shorter than that of CSLW-NBS.

\subsection{Generalization to Untrained Scenarios}
\label{app:unseen_scenario}
\begin{figure*}[t]
    \centering
    \begin{minipage}[t]{0.31\textwidth}
        \centering
        \includegraphics[width=\linewidth]{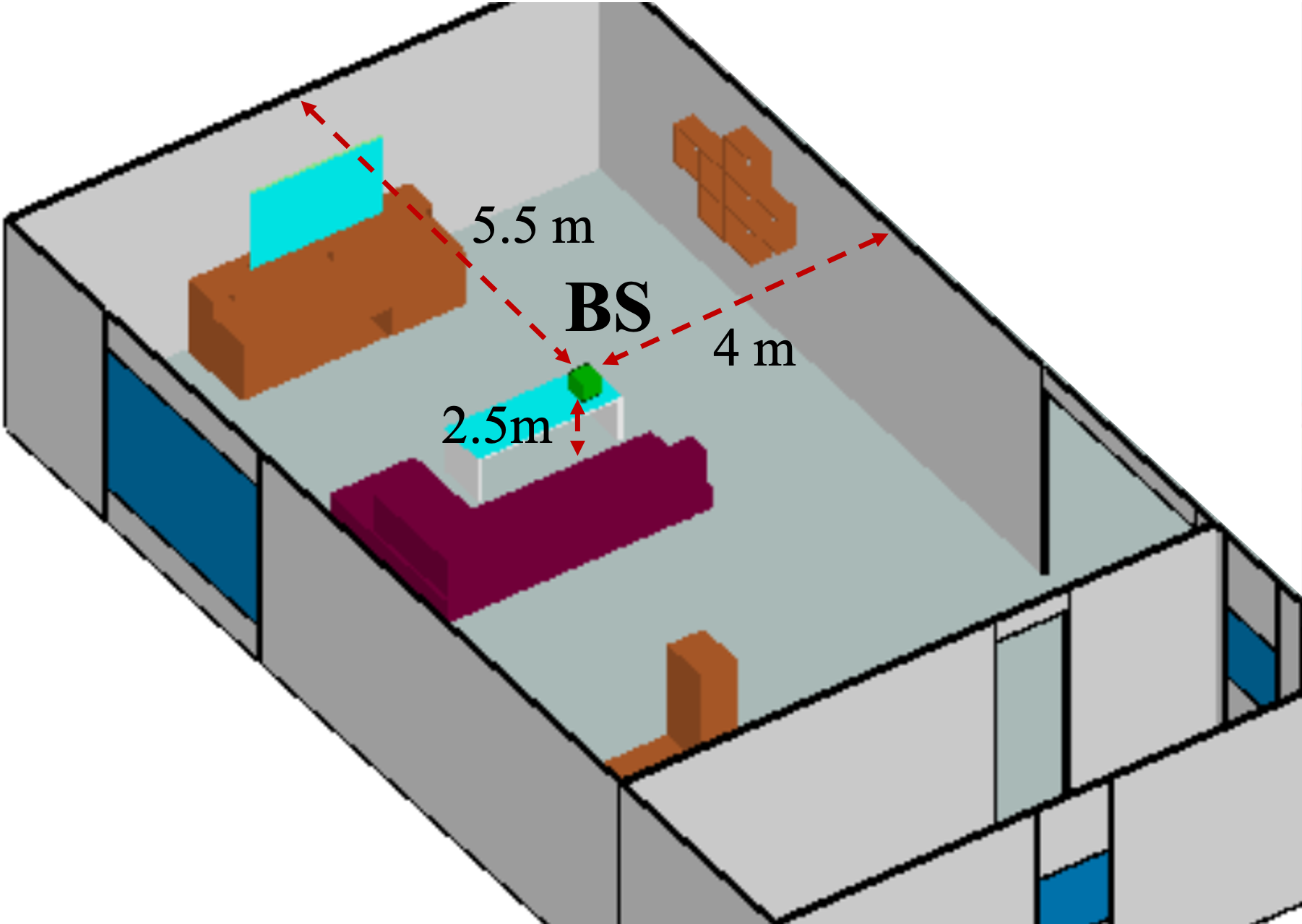}
        \caption{An illustration of indoor scenario ``I2''.}
        \label{fig:outdoor_scenario_I2}
    \end{minipage}%
    \hfill
    \begin{minipage}[t]{0.68\textwidth}
        \centering
        \includegraphics[width=\linewidth]{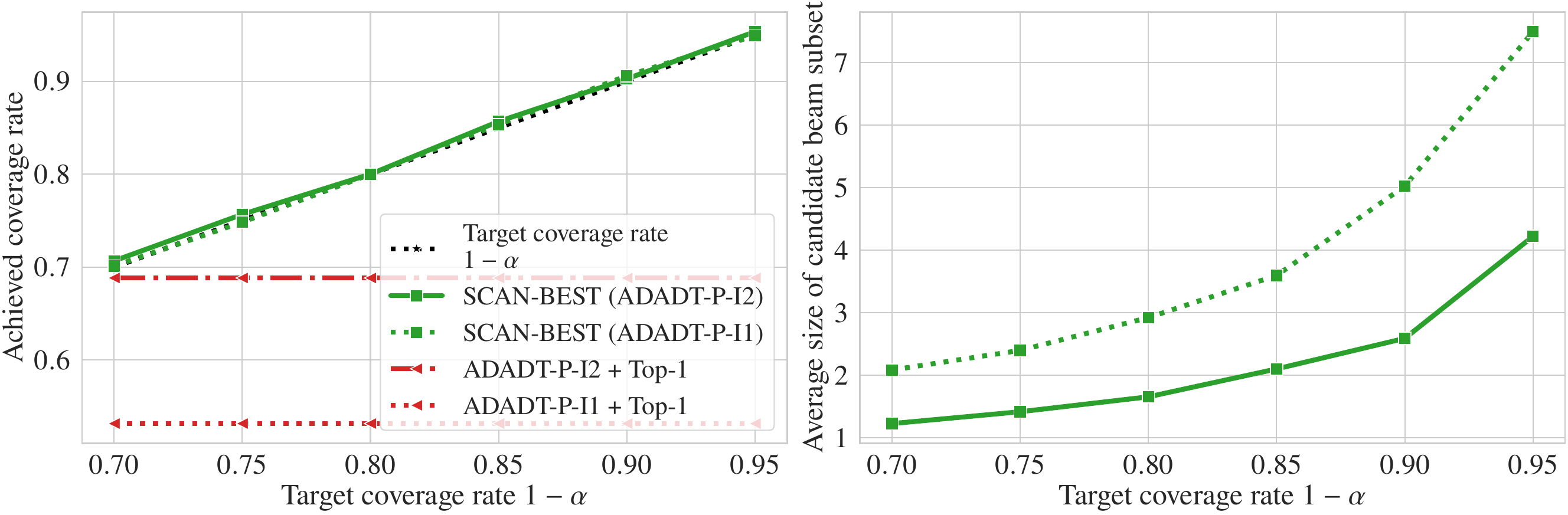}
    \caption{Achieved coverage rates and average sizes of candidate beam subset of SCAN-BEST and baselines in the untrained scenario I2.}
    \label{fig:coverage_and_size_untrained_I2}
    \end{minipage}
\end{figure*}
Here, we evaluate the generalization capability of SCAN-BEST (ADADT-P) trained under the scenario shown in Fig.~\ref{indorr simulation region}, denoted as SCAN-BEST (ADADT-P-I1), on an unseen scenario illustrated in Fig.~\ref{fig:outdoor_scenario_I2}. 
All other parameters remain consistent with Table~\ref{detailed parameters}. 
For comparison, we also train an ADADT-P model directly on the unseen scenario, referred to as ADADT-P-I2.
As shown in Fig.~\ref{fig:coverage_and_size_untrained_I2}, ADADT-P-I1 + Top-1 achieves comparable coverage rate to ADADT-P-I2 + Top-1 (77\%).
When integrated with the CRC framework, both SCAN-BEST (ADADT-P-I1) and SCAN-BEST (ADADT-P-I2) successfully achieve the desired coverage guarantees across different target coverage rates, while maintaining candidate beam sets of similar sizes.

\subsection{Extended Evaluation to Outdoor Scenarios}
\label{app:outdoor_scenario}
\begin{figure*}[t]
    \centering
    \begin{minipage}[t]{0.31\textwidth}
        \centering
        \includegraphics[width=\linewidth]{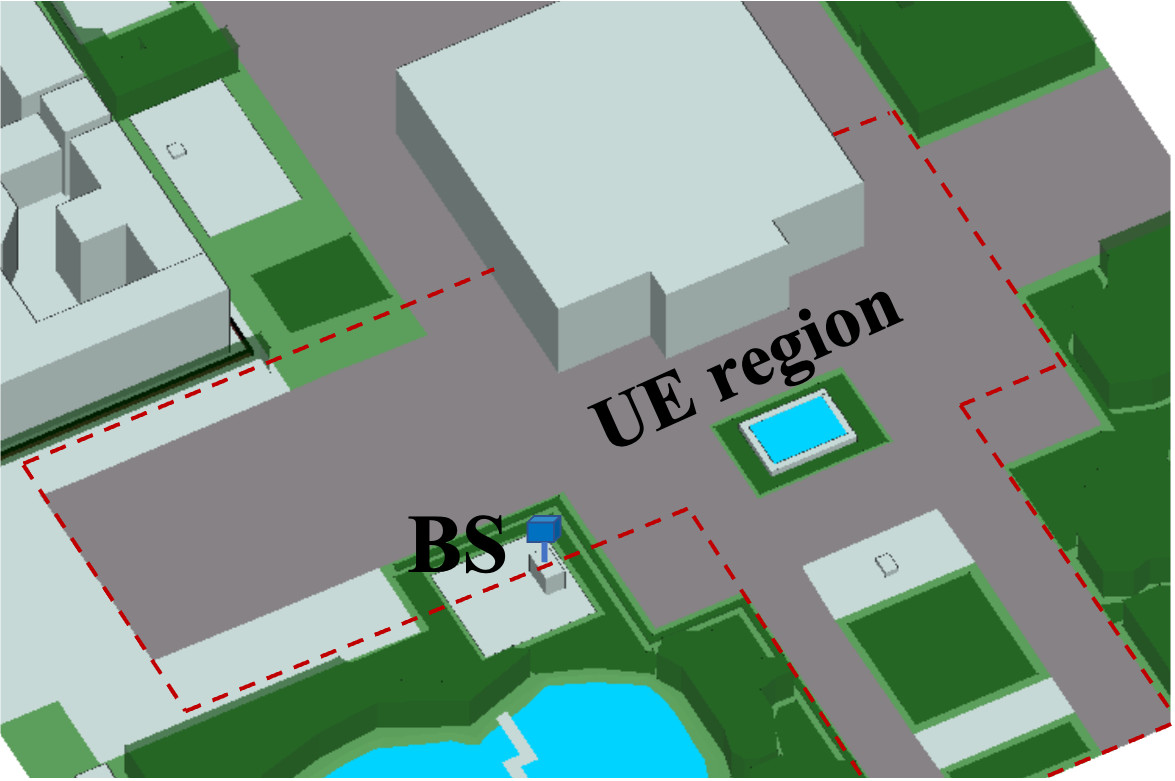}
        \caption{An illustration of outdoor scenario.}
        \label{fig:outdoor_scenario}
    \end{minipage}
    \hfill
    \begin{minipage}[t]{0.68\textwidth}
        \centering
        \includegraphics[width=\linewidth]{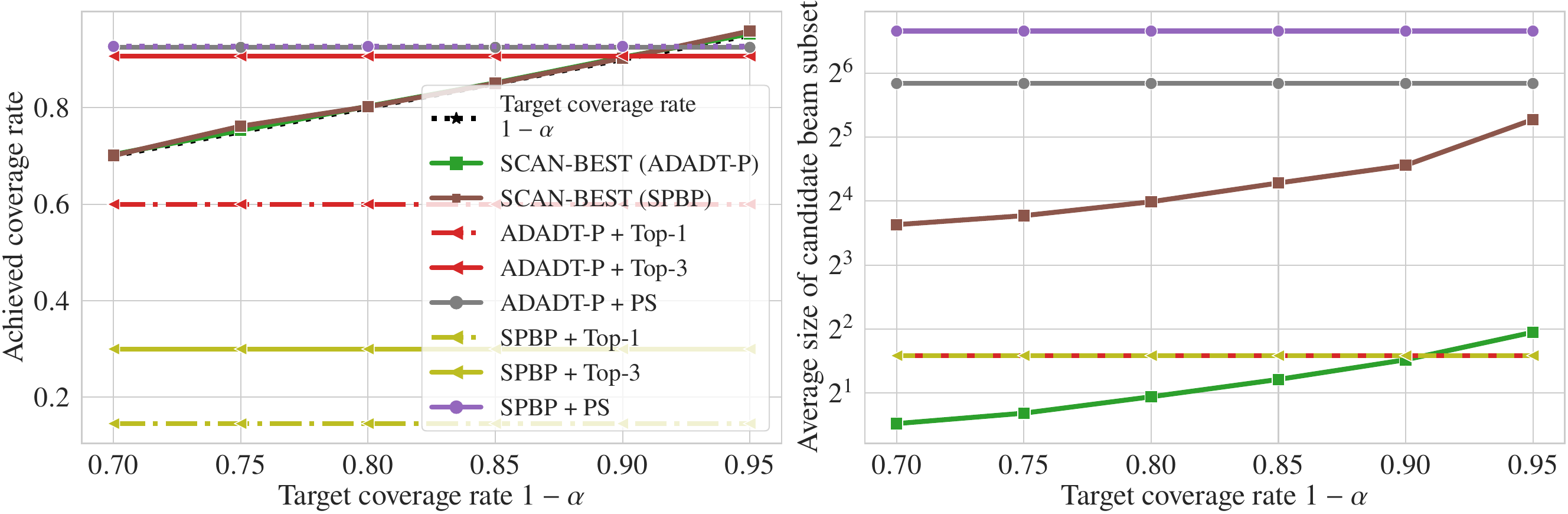}
        \caption{Achieved coverage rates and average sizes of candidate beam subset of SCAN-BEST and baselines in the outdoor scenario.}
        \label{fig:coverage_and_size_outdoor}
    \end{minipage}
\end{figure*}
Considering an outdoor scenario illustrated in Fig.~\ref{fig:outdoor_scenario}, the BS height is set to 15~m, the number of mmWave antennas at the BS is 512.
All other parameters remain consistent with those listed in Table~\ref{detailed parameters}. 
As shown in Fig.~13, both SCAN-BEST (ADADT-P) and SCAN-BEST (SPBP) efficiently achieve the desired coverage guarantees across various target coverage rates by dynamically forming their candidate beam sets, even in outdoor environment.
In contrast, the Top-$K$ and PS-based schemes fail to do so.
Moreover, the smaller candidate beam set obtained by SCAN-BEST (ADADT-P) relative to SCAN-BEST (SPBP) further demonstrates the superior efficiency of the employed ADADT.

\bibliographystyle{IEEEtran}
\bibliography{IEEEabrv,ref_simplified}
\end{document}